\documentclass[conference]{IEEEtran}
\IEEEoverridecommandlockouts
\usepackage{cite}
\usepackage{amsmath,amssymb,amsfonts}
\usepackage{graphicx}
\usepackage{textcomp}
\usepackage{algorithm}
\usepackage{algpseudocode}
\usepackage{mwe}
\graphicspath{ {images/} }
\usepackage{textcomp}
\usepackage[font=scriptsize]{subcaption}
\usepackage[font=scriptsize]{caption}
\usepackage{comment}
\usepackage{float}
\usepackage{soul}
\usepackage{hyperref}
\usepackage{xcolor}
\usepackage{amsthm}
\usepackage{wrapfig,lipsum,booktabs}
\usepackage{multirow,multicol}
\newtheorem{definition}{Definition}

\newtheorem{remark}{Remark}

\newtheorem{claim}{Claim}

\def\BibTeX{{\rm B\kern-.05em{\sc i\kern-.025em b}\kern-.08em
    T\kern-.1667em\lower.7ex\hbox{E}\kern-.125emX}}

\def\BibTeX{{\rm B\kern-.05em{\sc i\kern-.025em b}\kern-.08em
    T\kern-.1667em\lower.7ex\hbox{E}\kern-.125emX}}
    
\begin{document}

\title{Enhancing Attack Resilience in Real-Time Systems through Variable Control Task Sampling Rates}
\author{Arkaprava Sain, Sunandan Adhikary, Ipsita Koley, Soumyajit Dey\\
Department of Computer Science and Engineering, Indian Institute of Technology, Kharagpur, India
}

\maketitle

\begin{abstract}
    Cyber-physical systems (CPSs) in modern real-time applications integrate numerous control units linked through communication networks, each responsible for executing a mix of real-time safety-critical and non-critical tasks. To ensure predictable timing behaviour, most safety-critical tasks are scheduled with fixed sampling periods, which supports rigorous safety and performance analyses. However, this deterministic execution can be exploited by attackers to launch inference-based attacks on safety-critical tasks. This paper addresses the challenge of preventing such timing inference or schedule-based attacks by dynamically adjusting the execution rates of safety-critical tasks while maintaining their performance. We propose a novel schedule vulnerability analysis methodology, enabling runtime switching between valid schedules for various control task sampling rates. Leveraging this approach, we present the \emph{Multi-Rate Attack-Aware Randomized Scheduling (MAARS)} framework for preemptive fixed-priority schedulers, designed to reduce the success rate of timing inference attacks on real-time systems. To our knowledge, this is the first method that combines attack-aware schedule randomization with preserved control and scheduling integrity. The framework’s efficacy in attack prevention is evaluated on automotive benchmarks using a Hardware-in-the-Loop (HiL) setup.
\end{abstract}


\section{Introduction}
\label{secIntroduction}
\par \noindent Real-time systems (RTSs) in CPS domains like automotive, avionics, power grids, etc. are mostly safety-critical. 
The feature-rich implementation of such modern CPSs makes the system design mixed-critical, for which tasks with different criticality levels are executed in each processor~\cite{baruah2011response}. 
The safe operation of these RTSs relies highly on the timely execution of the safety-critical control tasks. 
For achieving predictable execution behaviour, fixed priority based scheduling algorithms remain popular for running tasks on embedded control units in CPSs, with higher priority typically being assigned to hard real time safety-critical tasks for ensuring deadline satisfaction even under platform uncertainties like sensing/actuation/communication delays. Also, the deterministic nature of static priority scheduling enables the designers to analyse the worst-case response time (WCRT) of every safety-critical control loop and design delay-aware control strategies for these loops~\cite{bini2008delay}.
\par However, this predictability of fixed-priority scheduling exposes the timing information of safety-critical task executions. The attacker first finds out the periodicity of a safety-critical task by observing its data communication in the network or by observing the changes in physical system states of the controlled plant ~\cite{prates2020defense,chen2019novelsidechannel}. By compromising a relatively lower priority task (mostly non-safety-critical) that periodically executes immediately before or after this victim task, the attacker observes its execution timestamps for multiple hyper-periods. 
Since a static fixed-priority schedule is followed, the execution sequence repeats every hyper-period, which helps the attacker infer possible initial offsets of the safety-critical task and predict their future arrival times. 
The attacker can manipulate data inside the victim task's shared I/O device buffer within a time window around those inferred future arrival time instances~\cite{chen2019novelsidechannel}. We term such vulnerable time windows as {\em attack effective window (AEW)}, a quantity which is implementation-dependent and can be measured experimentally by the attacker~\cite{chen2019novelsidechannel}.
Such schedule-based attacks (SBAs) can be categorised into the following four models depending on the timing relation between the victim's arrival and the attacker task's arrival within AEW~\cite{schedguard++}, i.e., Posterior Attack (attacker task executes after victim task), Anterior Attack (attacker task executes before the victim task), Concurrent Attack (attacker task executes while the victim task runs) and, Pincer Attack (combines both posterior and anterior attack). Since the non-critical attacker tasks are of lower priority, the corresponding jobs are mostly scheduled/completed after higher-priority victim task instances in case of nearby arrival times. This makes posterior attacks more prevalent compared to other attack models in embedded controllers~\cite{ren2023protection}. In this paper, we consider the posterior attack model, analyse the drawbacks of state-of-the-art defence mechanisms, and discuss how the proposed defensive solution outdoes them.
\par Researchers have proposed several solutions to prevent and defend against such SBAs. Among these, \emph{schedule randomization}-based defence approaches are widely explored. One of the initial works in this area is \textit{TaskShuffler}~\cite{yoon2016taskshuffler}, which randomizes job execution by allowing lower-priority jobs in the ready queue to run before higher-priority jobs. 
In~\cite{kruger2018vulnerability}, the authors discuss the following two {\em slot-level randomization techniques} to mitigate directed schedule-based attacks. {\em(i)} Slot-level online randomization of schedules, for selecting the next job at runtime meeting task deadline constraints and {\em (ii)} Offline schedule-diversiﬁcation, where offline precomputed schedules are stored beforehand, and their selection is made online. Such schedule randomization methods obfuscate the schedule information, which makes it difficult for the attacker to predict the future arrival times of the victim task.
However, randomizing without any insight into the attacker's strategy produces some vulnerable schedules during randomization, which can make a random attack attempt successful~\cite{nasri2019pitfalls}. The authors in~\cite{chen2021indistinguishability} devise a privacy-preserved strategy that samples a suitable execution offset for critical tasks from a Laplace distribution and switches between them to keep the execution sequences from getting exposed. However, they do not ensure the performance and schedulability criteria of the control tasks while doing so. In another line of work, the authors in~\cite{schedguard++} suggested a temporal protection mechanism for multi-core Linux platforms. In this method, all the non-safety critical tasks, mostly considered to be {\em untrusted}, are blocked within the \emph{AEW} of all the safety-critical or {\em trusted} tasks, eliminating all schedule vulnerabilities. However, this significantly delays the execution of {\em untrusted} tasks in uniprocessor systems, and they frequently miss their deadlines. If the trusted or critical task utilization is considerably high, such a protection strategy can impose an unbounded delay on the execution of other tasks that hampers the overall system  performance~\cite{ren2023protection}. In this work, we propose a schedule randomization policy that is attack-aware and ensures meeting the deadlines for all tasks in the system, by carefully balancing the trade-offs between performance and security. 
%
\par One of the crucial pieces of information for any schedule-based attack model is the future arrival instances of a critical task, which is derived by finding out its periodicity~\cite{chen2019novelsidechannel}. If the sampling period is dynamically switched, the attacker will find it difficult to infer the exact schedule even after doing schedule analysis. With this motivation, we utilise multi-rate controller scheduling strategy to ensure security in real-time task schedules by reducing its \emph{inferability},  i.e. identifying start times of victim task instances. 
As shown in the state-of-the-art research works, dynamic switching between a predefined set of sampling periods, coupled with optimally designed controllers, significantly reduces resource utilization, ensuring performance, compared to conventional fixed sampling rate static scheduling~\cite{henriksson2005optimal}. 
There exist state-of-the-art works that analyse the performance of such systems by analysing their Lyapunov stability~\cite{schinkel2002optimal,adhikary2022work,sain2023work}. Such works mostly analyse whether a \emph{Common Quadratic Lyapunov Function (CQLF)} exists among the control loops with different sampling rates. The existence of such a function that abides by a given performance criteria guarantees stability when the system is subject to arbitrary or fast switching.
There are works that derive a \emph{Multiple Lyapunov functions (MLF)}-based~\cite{liberzon1999basic} or \emph{dwell time}-based switching rules for slowly switching between different modes to ensure stability. Most of these research works focus on optimizing resources or control costs utilizing multi-rate switching strategies~\cite{zamani2016scheduling,bini2008delay, moraes2014sampling, schinkel2001sample}. In this work, we intend to switch between different controller sampling rates in runtime to reduce the \emph{inferability} and {\em vulnerability} of a deployed schedule towards posterior SBAs. To enable this arbitrary switching, we utilize the CQLF-based fast-switching rules that ensure stability.
%
\par We propose a novel \emph{ Multi-Rate Attack-Aware Schedule Randomization (MAARS) Framework} to secure all the safety-critical control tasks that are scheduled to execute in a uni-processor embedded platform against schedule-based posterior attacks~\cite{nasri2019pitfalls}. This essentially combines {\bf (I)} an \emph{offline schedule diversification} technique that chooses a set of schedules with minimum {\em inferability}/predictability for a schedule-based attacker and {\bf(II)} an {\em online schedule selection} technique that suitably deploys schedules to minimize the {\em vulnerability} of an attack-affected safety-critical task. 
%
%
%
To the best of our knowledge, this is the first work that proposes an attack-aware schedule randomization policy. The following are the main contributions of this work:
\par\noindent \textbf{1. }We present a novel approach for judiciously choosing performance-aware sampling rates for a safety-critical control task such that the \emph{inferability} of its future arrival instances are minimized, thereby effectively thwarting the leakage of critical task information \emph{(such as periodicity, initial task offset, etc.)}
\par\noindent \textbf{2. }We propose a novel methodology that generates a set of schedules offline using slot-level randomization techniques for the chosen set of sampling rates for each control task. We also develop a {\em runtime schedule selection algorithm} that intelligently deploys schedules from this set by observing the posterior attacks' effect on a control loop 
such that the vulnerability of the corresponding control task is minimized in that schedule. 
\par\noindent \textbf{3. }We experimentally evaluate how well the proposed \emph{MAARS} framework balances the trade-off between security and system performance compared to state-of-the-art techniques by applying it to a synthetic automotive task set.
We implement the same task set in a \emph{Hardware-in-Loop (HIL)} testbed setup to judge the \emph{real-time} applicability of the MAARS framework. 

\section{System Model}
\label{sec:overall_system_model}
\subsection{Task and Scheduler Model}
\par \noindent Let  $\Gamma$ denote a set of N periodic tasks s.t. $\Gamma = \{\tau_1, \tau_2, ... \tau_N  \}$, that are scheduled on a uniprocessor system by a static priority scheduler with fixed priority. Each task has three parameters $(e_i, p_i, d_i)$, where $e_i$ is the worst-case execution time (WCET), $p_i$ is the sampling period, and $d_i$ is the relative deadline. We assume $d_i = p_i,\, \forall i\leq N$ and the WCET of a task includes jitter and context-switch overheads. We consider a mixed-criticality system, i.e., it comprises control tasks of varying criticality levels in order of their priority. 
Following the work in \cite{schedguard++}, we consider all the safety-critical control tasks as trusted and hard to compromise. Whereas other tasks are considered to be untrusted and can be compromised by a schedule-based attacker. 
Any safety-critical control task from the \emph{trusted} set $\Gamma_t
\subseteq \Gamma$ can be a victim to a \emph{schedule-based attack} (SBA) launched by an adversary who has compromised a non-safety-critical task from the \emph{untrusted} set $\Gamma_u \subseteq \Gamma$. Due to the high criticality level, the trusted tasks are mostly assigned with higher priority than the untrusted tasks. Hence, their executions are not delayed by the non-critical tasks as they are promptly scheduled as they arrive by a {\em fixed-priority preemptive scheduler}. In this work, we consider these trusted safety-critical control tasks to be designed with different sampling rates, which changes their periodicities of execution. We denote the set of periodicities for $i^{th}$ control task as $P_i=\{p1_i,p2_i,\cdots,pl_i\}$. We assume that the task set $\Gamma$ is schedulable by a fixed-priority preemptive scheduler such that the \emph{worst-case response time (WCRT)}  $wcrt_i$
for any task $\tau_i \in \Gamma$ satisfies,
  $\scriptstyle  wcrt_i = e_i + \sum\limits_{\tau_j \in hp(\tau_i)} \left\lceil \frac{r_i^k}{p_j} \right\rceil e_j \leq d_i$. 
Here $r_i^0 = e_i$, $wcrt_i = r_i^{k+1} = r_i^k$ for some $k\geq 0$ and for every $\tau_i\in\Gamma_t$, $p_i=$ minimum periodicity from the set $P_i$. This criteria ensures that the tasks are {\em schedulable} for their maximum execution rates and assuming no precedence constraints among the tasks,  {\em valid schedules} can be generated without any deadline misses. 
%
\subsection{Control Task Model}
\label{sec:controlloopdetectormodel}
\par \noindent We express the plant model as a \emph{Linear time-invariant (LTI)} system having dynamics as follows.
\begin{align}
\label{eq:system_equation}
\nonumber
    &\dot{x}= A_c x(t) + B_cu(t), \ y(t)= Cx(t)+v(t)\\
\nonumber
    &\hat{x}[k + 1] = (A-LC) \hat{x} [k] + B u[k] + L y[k]\\ &u[k+1] = K \hat{x} [k+1]  
\end{align}
Here, the vectors $x \in \mathbb{R}^n, \hat{x} \in \mathbb{R}^n, y \in \mathbb{R}^m$, and $u \in \mathbb{R}^p$ describe the plant state, the estimated plant state, output and the control input, respectively. The matrices $A_c$, and $B_c$, are continuous time matrices that define the continuous-time state and input-to-state transition matrices, respectively. $C$ is the output transition matrix. At the $k^{th}$ sampling instance, the real-time is $t = kh$, where $h$ is the sampling period. We can, therefore, define the discrete-time matrices corresponding to their continuous-time counterparts as $A= e^{A_c h}, \ B= \int_{0}^{h} e^{A_c t}B_c \cdot dt $. $K$ and $L$ are Linear Quadratic Regulator (LQR) controller gain and Kalman estimator gain, respectively. Each control task samples the sensed plant measurement vector $y$ at each sampling iteration, using which future states of the plant are estimated. Based on the estimated plant state, a control task calculates the required control input $u$ to transmit it to the actuator to actuate and control the plant output.
\begin{figure}[!ht]
\label{Fig:Task_Model}
\centering
\vspace{-3mm}
{\includegraphics[clip,width=0.98\columnwidth]{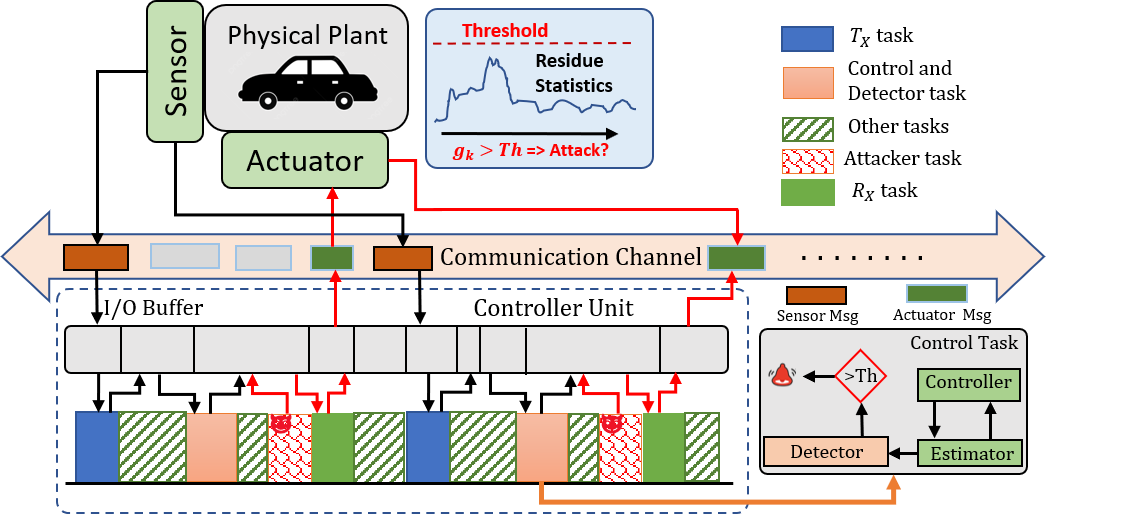}}
\caption{System and Control Task Model}
\vspace{-3mm}
\end{figure}
\subsubsection{Closed Loop Under Sampling Period Variation}
\label{sec:multi_rate switching}
As discussed above, the discrete-time system characteristic matrices depend on the sampling period $h$. Since the LQR gain is designed to minimize a user-defined quadratic cost function, which is sampled once in every $h$ time interval, the LQR control gain is also dependent on the sampling period. So is the estimator Kalman gain $L$, as it is also designed to minimize the variance of the estimation error, i.e., the difference between the estimated plant state $\hat{x}$ and actual plant state $x$ at every sampling iteration.
Therefore, the discrete-time matrices $A, B, K, L$ in Eq.~\ref{eq:system_equation} can be replaced with their sampling period-dependent versions, i.e., $A_h,B_h,K_h,L_h$. By discretizing the continuous-time plant equation in Eq.~\ref{eq:system_equation} we get $x[k+1] = A_hx[k] + B_hu[k]$. Since we intend to capture the change in the progression of the control loop w.r.t. the change in its sampling period,
we define an extended augmented state vector of the control task as $X = [x^T, \hat{x}^T]^T$ that captures both plant and estimator/controller dynamics. By replacing $u$ with $-K_h\hat{x}$ and $y$ with $Cx$, the overall closed-loop system evolves as follows.

{\small\begin{align} 
\label{eq:state_space}
     &X[k + 1] = \mathbb{A}_h X[k], \mathbb{A}_h = \begin{bmatrix} A_h & -B_hK_h \\ L_hC & A_h-L_hC-B_hK_h \end{bmatrix}
\end{align}}%
For two different sampling periods $p1$ and $p2$, we denote the discrete-time augmented system matrices by $\mathbb{A}_{p1}$ and $\mathbb{A}_{p2}$. Our methodology uses this augmented system representation in Eq.~\ref{eq:state_space} to ensure a desired performance while varying periodicities of control tasks for security.
\subsubsection{Residue-based Anomaly Detector}
\label{sec:detector_model}
Each of the safety-critical control loops is equipped with a residue-based detector that observes changes in the system residue (i.e., the difference between the sensed and estimated outputs) for anomaly detection. We consider an integrated implementation of controller and detector functionalities as part of the control task to avoid data manipulation. We use a $\chi^2$-based detector in this work. The $\chi^2$ detector utilises a normalized quadratic function of the residual to amplify and easily detect minute variations in system residue. For system residue $res[k]$ at $k^{th}$ sampling iteration, its chi-square measure is $z[k] = res[k]^T \Sigma_{res}^{-1} res[k]$ where $\Sigma_{res}$ is the variance of system residue. Considering the measurement noise to be a zero-mean Gaussian noise, $res[k] \sim \mathcal{N}(0,\Sigma_{res})\Rightarrow z[k]\sim \chi^2(m,2m)$, where $m$ is the number of output measurements, considered as the degree of freedom of the $\chi^2$ distribution. We employ a windowed chi-square detector that compares the average value of the chi-square statistic of system residue over a pre-defined time window ($N$), i.e., $g[k]=\frac{1}{N}\sum\limits_{i=0}^{N-1}z[k-i]$ and compares it with a pre-defined threshold $Th$. The threshold is calculated to maintain a desired false alarm rate~\cite{koley2021catch}. The chi-squared detector raises the alarm, denoting an attack attempt on a certain closed loop when $g[k]> Th$ at any $k^{th}$ sampling period of that closed loop.
\section{Threat Model}
\label{secthreatmodel}
In this section, we discuss in-depth how we model the schedule-based attacker and explain our assumptions about the attacker's objectives and capabilities. The objective of the attacker is to utilize the timing information exposed by the processor-level task execution schedules to compromise tasks with higher criticality. 
Since tasks with higher criticality are mostly developed by OEMs or {\em trusted} vendors and go through thorough security scans and functionality checks, they are usually very difficult to compromise~\cite{ren2023protection,schedguard++}. Hence, an attacker is more likely to attempt compromising less critical tasks, that are often assigned lower priorities. We denote such compromisable tasks as {\em untrusted} tasks. Since these untrusted tasks are part of the software that comes from a range of third-party vendors\cite{ren2023protection}, compromising them using a man-in-the-middle-type attack is possible. As mentioned in Sec.~\ref{secIntroduction}, depending on the timing relationship between the victim and attacker task, SBAs can be categorised into 4 types~\cite{schedguard++}. 
This work considers a posterior SBA model and provides a mitigation framework against it. The following are assumptions regarding the attacker's ability, objectives, and attack methods.
\subsubsection{Attacker's Capabilities}\label{secAtkassum}
We make the following assumptions about the attacker's capabilities: \textit{(1)} The scheduler is trustworthy and can not be compromised by the attacker. Consequently, the attacker lacks direct control over task execution and is restricted to leveraging compromised tasks only when the scheduler permits.
\begin{wrapfigure}[8]{l}{0.42\columnwidth}
    \centering   
    \vspace{-4mm}
\includegraphics[width=0.44\columnwidth,clip]{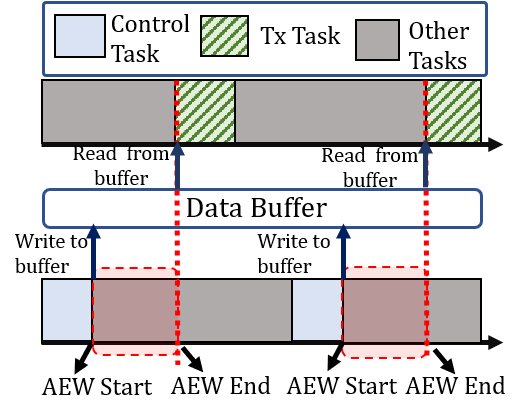}
\vspace{-7mm}
   \caption{Attack Effective Window}
        \label{fig:AEW_example}
\end{wrapfigure}
\textit{(2) } The adversary knows the scheduling policy used, and it can compromise \emph{one task among all the} lower-priority untrusted tasks ($\Gamma_u$). This enables the attacker to observe executions of this compromised untrusted task for several hyper-periods to find out partial information about its intended victim control task($\in\Gamma_t$). \textit{(3) }The attacker can partially deduce the sampling rates of any control task ($\in\Gamma_t$) either by observing the message data packets in the network or physical system states.
\textit{(4)} Attacker can modify the data written by the victim control task inside an I/O buffer or cache only if it executes within an \emph{Attack Effective Window}~\cite{schedguard++}. AEW is the specific time interval within which the attacker task must be executed for successful data manipulation. This duration is completely implementation-dependent and can be derived by experiments~\cite{chen2019novelsidechannel}. The attacker must tamper the data in the I/O buffer before the transmission task reads and transmits it for actuation. Hence, AEW is the duration after which the transmission task transmits the control program-generated actuation data and is scheduled to execute with the same periodicity as the control task. In Fig.~\ref{fig:AEW_example}, we illustrate this where the AEW of the control task has been shaded in red. AEW starts when the control task (in blue shade) finishes execution, writes its data to the buffer and ends once the transmission task (in green stripe) starts.
We denote the AEW length for a control task $\tau_i \in \Gamma_t$ by $\Omega_i$.
%
\subsubsection{Attack Execution} \label{secAtkexec}
We assume an attacker's objective is to manipulate control inputs by performing a schedule-based posterior attack on control tasks. To implement such an attack, the attacker must infer the exact future arrival instances of the victim task. Having compromised a lower-priority untrusted task, the attacker can log that task's start and end time as part of the overall schedule. Utilising this information, it attempts to predict the start and end times of other trusted tasks in the schedule. The attacker can achieve this by using a \emph{schedule ladder-based analysis} proposed in the work~\cite{chen2019novelsidechannel}. We briefly discuss how the schedule ladder is constructed, and the attacker uses it to derive the future arrival instances of the victim task.
\begin{wraptable}[5]{l}{0.30\columnwidth}
\captionsetup{font=scriptsize}
    \centering
    \scriptsize
    \begin{tabular}{|c||c|c|}
        \hline
        Tasks & $P$ & $e$\\
        \hline\hline
        $\tau_1$ & \{4,5\} & 1\\
        \hline
        $\tau_2$ & 4 & 1\\
        \hline
        $\tau_3$ & 5 & 2\\
        \hline
    \end{tabular}
    \caption{\centering Example Task-Set 1}
\label{tab:motivating_task_set_1}
\end{wraptable}
A schedule ladder represents a task schedule arranged vertically with adjacent timelines of equal row size. Each column in the ladder represents an atomic observation period, or unit time, of real-time duration $\delta$. The attacker takes the row size of the ladder as equal to the victim task's sampling period measured in time units of size $\delta$. Hence, for a given victim task $\tau_i$ executing with a sampling rate $p_i$ time units, the row length of the ladder becomes $\delta p_i$. The motive behind taking such a row length is that the victim task will arrive once in every row of the ladder and will occupy $\delta e_i$ columns. Moreover, all arrivals of $\tau_i$ will be located in the same column. We illustrate a schedule ladder diagram for the task set in Tab.~\ref{tab:motivating_task_set_1} in  Fig.~\ref{fig:ladder_analysis}. The row length of the ladder is 4 since $\tau_1$(victim) sampling rate = 4. The first row represents the timeline $[0,4]$, second row $[4,8]$ and so on. The blue arrow shows the arrival column of the victim task that arrives once in every row and occupies column 1 of the ladder.  
While performing an analysis, the starting point of analysis (observation) of the top row can be arbitrarily chosen by the attacker, which is decided by the start of its observation window. The attacker starts observing its own start and end times using the compromised untrusted task. Note the attacker task is a lower priority untrusted task (w.r.t victim). Hence, if the victim task arrives in the same column, it preempts the attacker task. This is shown in Fig.~\ref{fig:ladder_analysis}, where the attacker task $\tau_3$ arrives at instances $t=0,20$ in the first column (since its sampling rate is 5) but gets preempted by $\tau_1$. In column 2 it gets preempted by $\tau_2$. Therefore, it always executes in the third and fourth column of the ladder.
\begin{wrapfigure}[13]{l}{0.41\columnwidth}
\vspace{-4mm}
{\includegraphics[clip,width=0.40\columnwidth]{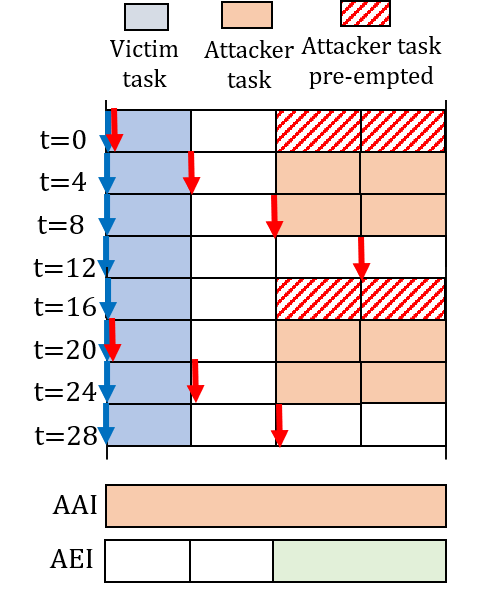}}
\caption{ \centering Ladder diagram of Task-Set-1 illustrated along with their corresponding AEI and AAI.}
\vspace{-5mm}
\label{fig:ladder_analysis}
\end{wrapfigure}
%
\par After the end of observation, the attacker finds out \emph{Attack Arrival Instances (AAI)}. The AAI is the set of ladder columns where the attacker task has arrived at least once during its entire observation window. For e.g. in Fig.~\ref{fig:ladder_analysis}, the AAI is $[0,4]$, shown below the ladder diagram. The set of ladder columns where the attacker task has been executed at least once during its entire observation window are termed as \emph{Attack Execution Instances (AEI)}. Note that the ladder columns where the attacker task arrives ($\in AAI$) but never gets executed ($\notin AEI$ since it is preempted by higher priority victim tasks) are the possible candidates for the columns where the victim task arrives. Hence, the attacker easily infers these ladder columns, which are in AAIs but not AEIs, as the victim task's arrival column (offset), predicting them as the victim's future arrival instances. We illustrated this in Fig.~\ref{fig:ladder_analysis} for task-set Tab.~\ref{tab:motivating_task_set}, where using $\tau_2$ as an attacker task results in $AAI=4$ and $AEI=2$. One immediate conclusion we draw from this discussion is that a longer execution time for the attacker task make it easier for an attacker to infer the arrival column of the victim task. This is because the attacker task execution will span a greater proportion of ladder columns, making it easier to guess the correct victim task arrival column.
\par Consider the task set in Tab.~\ref{tab:motivating_task_set}, where the trusted control task set $\Gamma_t =\{\tau_1, \tau_2\}$ and the untrusted task set $\Gamma_u =\{\tau_3\}$ (Fig.\ref{fig:Attack_Effective_Window}). After determining the arrival time of the victim task $\tau_1$, the attacker uses untrusted task $\tau_3$ to launch a posterior SBA on $\tau_1$. We assume the AEW of $\tau_1$ is $\Omega_1 = 1$ time unit (highlighted in the shaded region of the timeframe). Therefore, for a successful SBA, $\tau_3$ must execute within this timeframe after $\tau_1$ finishes executing. Using this threat model, we demonstrate how state-of-the-art secure scheduling policies can mitigate such attacks and discuss their shortcomings in balancing performance and security, which motivates our novel methodology.
%
\vspace{-1mm}
\section{Motivating Example}
\label{sec:motivating_example}
%
\noindent We motivate our proposed methodology with an example by discussing how {\em (i)} the conventional protection window-based secure scheduling methods, such as~\cite{schedguard++}, fail to ensure system performance, optimal resource utilisation, etc., and {\em (ii)} the randomization-based scheduling methods, such as~\cite{yoon2016taskshuffler}, fail to guarantee security against SBAs. This, in turn, motivates our methodology, which utilises variable sampling period-based implementation of control tasks for balancing the trade-offs between security and performance without compromising the schedulability of the task set.
\begin{wraptable}[5]{l}{0.38\columnwidth}
\captionsetup{font=scriptsize}
    \centering
    \scriptsize
    \vspace{-3mm}
    \begin{tabular}{|c||c|c|}
        \hline
        Tasks & $P$ & $e$\\
        \hline\hline
        $\tau_1$ & \{2,3\} & 1\\
        \hline
        $\tau_2$ & 4 & 1\\
        \hline
        $\tau_3$ & 4 & 1\\
        \hline
    \end{tabular}
    \caption{Example Task Set-2}
    \label{tab:motivating_task_set}
\end{wraptable}
\noindent We refer the task set in Tab.~\ref{tab:motivating_task_set}, where the trusted control tasks $\Gamma_t =\{\tau_1, \tau_2\}$ and the untrusted tasks $\Gamma_u =\{\tau_3\}$. The control task $\tau_3$ executes at two different sampling rates $2$ and $3$ time units ($\delta$ time unit$=10 ms$), respectively. We assume the attacker has compromised the untrusted task $\tau_3$ to launch a posterior attack on the victim task $\tau_1$. The AEW duration for $\tau_1$ is $\Omega_1=1$ time unit. 
\begin{wrapfigure}[7]{l}{0.55\columnwidth}
\centering
\vspace{-2mm}
{\includegraphics[width=0.57\columnwidth,clip]{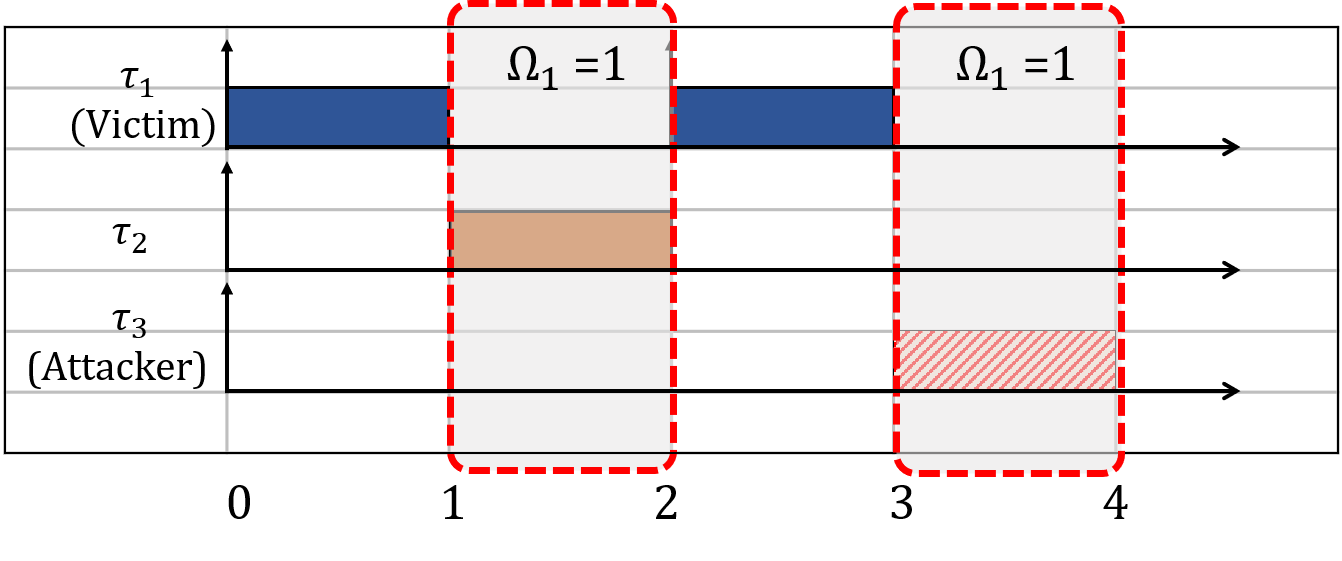}}
\caption{AEW-based scheduling of Taskset-2}
\label{fig:Attack_Effective_Window}
\end{wrapfigure}
\par Protection window-based approaches such as in~\cite{schedguard++} propose restricting the execution of untrusted tasks' job instances within the AEW of all the trusted tasks. 
Fig.~\ref{fig:Attack_Effective_Window} demonstrates how a task schedule generated following~\cite{schedguard++} (for the task set in Tab.~\ref{tab:motivating_task_set}) leads to deadline misses of untrusted tasks and wastage of CPU resources. The first job of $\tau_1$ executes at $[0,1]$, In the next interval at $[1,2]$, job instance of $\tau_2$ executes. In $[2,3]$, second job of task $\tau_1$ arrives and executes. However, the job instance of $\tau_3$ will not be allowed to execute at $[3,4]$ and hence will miss its deadline, even though the CPU remains idle. This leads to inefficient CPU utilisation.
\par On the other hand, the schedule randomization methods are 
more suitable since they make sure every task meets its deadlines (irrespective of trust). However, the state-of-the-art randomization policies do not consider the attacker's model into consideration and hence are often more vulnerable to SBAs while randomly changing the schedules~\cite{nasri2019pitfalls}.
Earlier works~\cite{yoon2016taskshuffler} have demonstrated that selecting schedules using such attack-unaware randomization policies can still lead to successful attacks. 
\begin{wrapfigure}[12]{l}{0.6\columnwidth}
\centering
\vspace{-5mm}
{\includegraphics[width=0.62\columnwidth,clip]{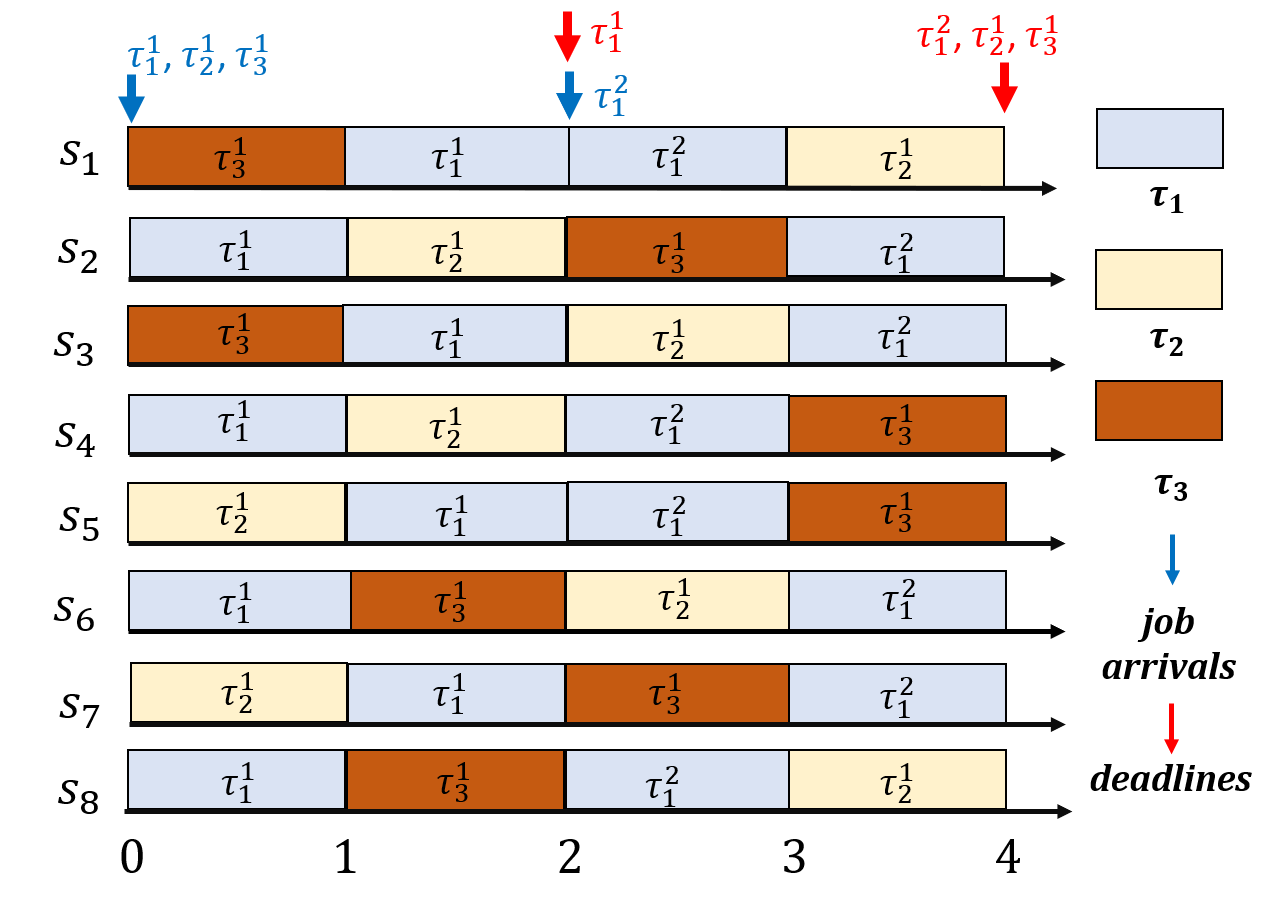}}
\caption{\centering \scriptsize Feasible Schedules Generated using Task-Set 2 (Considering $p_1=2$)}
\vspace{-12mm}
\label{Fig:feasible_schedule_examples}
\end{wrapfigure}
For example, if we apply such a schedule randomization method~\cite{yoon2016taskshuffler} on the same task set for a single sampling rate of $p_1=2$, for control task $\tau_1$, we can generate a total of $8$ unique feasible task schedules of hyper-period length of $4$ time units (see Fig.~\ref{Fig:feasible_schedule_examples}). But among them, 4 schedules $s_4,s_5,s_6$ and $s_7$ are vulnerable for $\tau_1$, since instances of $\tau_3$ appears within the AEW of the victim task $\tau_1$. However, if we use two different sampling periods for the control task $\tau_1, \ P_1=\{2,3\}$, we can form two possible task specifications $\Gamma_1$ and $\Gamma_2$ with unique sets of sampling period, i.e., $\{p_1=2,p_2=4,p_3=4\}$ and $\{p_1=3,p_2=4,p_3=4\}$. For the task specification $\Gamma_2$, we can generate $36$ unique and feasible schedules (using TaskShuffler) with a hyper-period length of $12$ time units. By analysing the vulnerability of $\tau_1$ for all 36 schedules, $12$ schedules are found to be safe (no posterior attackable instances within AEW) for victim task $\tau_1$. Thus, multi-rate execution enables the generation of a higher number of randomized schedules with minimal vulnerability.
From these examples, we draw a conclusion that utilising multi-rate for trusted tasks along with imparting attack awareness to the task schedules (w.r.t trusted tasks) can 
hamper a schedule-based posterior attacker's 
 success rate. 
\section{Proposed Methodology}
\label{sec:Methodology}

In this section, we present the methodology for generating multi-rate attack-aware randomized schedules. 
%
\begin{figure}[!hb]
\centering
\vspace{-4mm}
{\includegraphics[width=\columnwidth,clip]{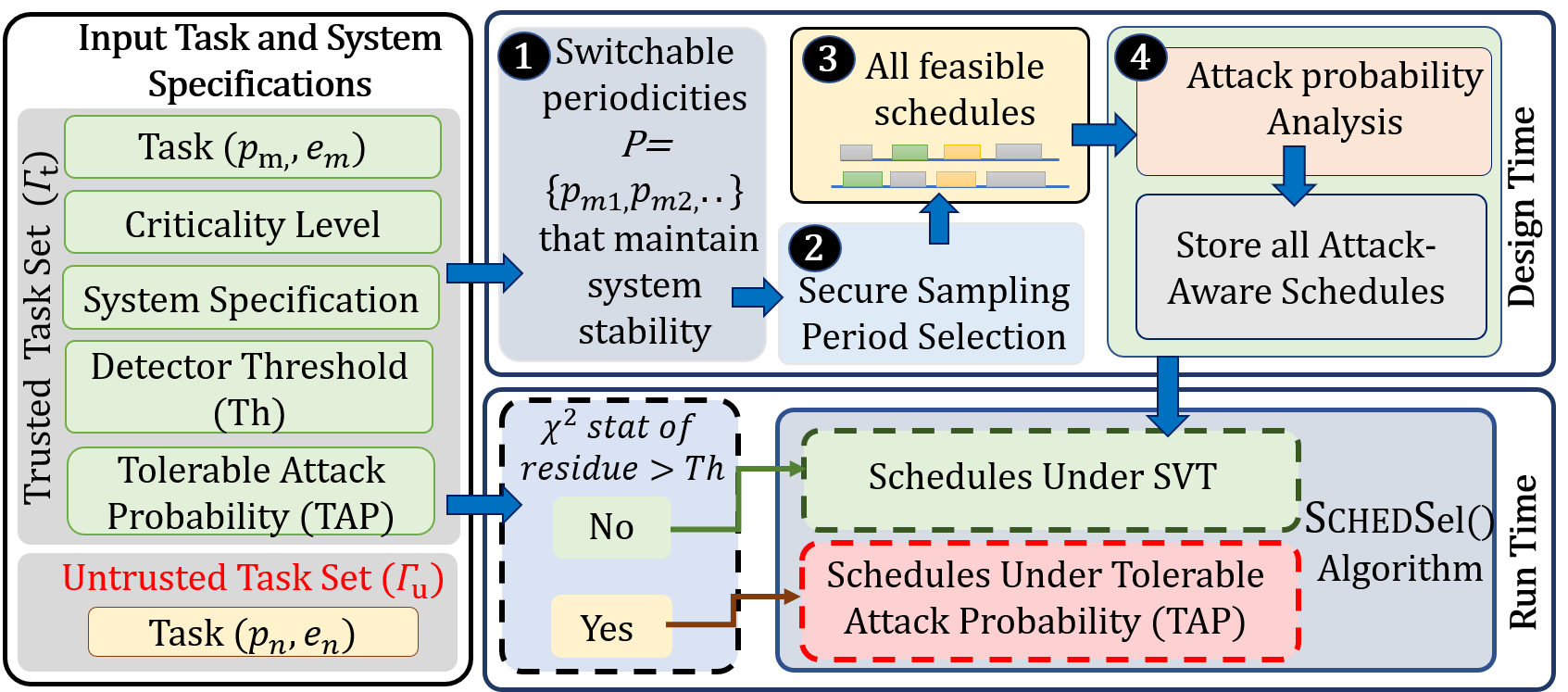}}
\caption{Overview of the proposed MAARS Framework}
\vspace{-4mm}
\label{fig: Overview_figure}
\end{figure}
A step-wise overview of the MAARS methodology is shown in Fig.~\ref{fig: Overview_figure}. During the design time, our proposed framework takes the following as inputs: {\bf (i)} a task set $\Gamma$ consisting of both trusted and untrusted set of tasks, $\Gamma_t,\Gamma_u (\subseteq \Gamma)$ respectively, {\bf (ii)} trusted task specifications, i.e., criticality level/ priorities $c_i$, a {\em controllable, schedulable} set of periodicities $P_i$, WCRT $e_i$ of each trusted task $\tau_i\in \Gamma_t$, {\bf (iii)} trusted task specifications, i.e., criticality level or priority $c_j$, periodicity $p_j$, WCRT $e_j$ of each untrusted task $\tau_j\in \Gamma_u$, {\bf (iv)} physical system specifications corresponding to each control task $\tau_i\in \Gamma_t$, its performance specification, {\bf (v)} other tunable design parameters related to MAARS, like \emph{Tolerable Attack Probability}, \emph{Maximum Schedule Vulnerability Index}, etc. to constrain the vulnerability of MAARS as defined in Sec.~\ref{section:schedule_analysis}

As motivated in Sec.~\ref{sec:motivating_example} switching between the allowable set of sampling periods $P_i$ reduces the determinism (explained as {\em inferability} in Sec.~\ref{subsec:Secure_Sampling_Period_DIscussion}) in the real-time task execution pattern. The {\em first} step in our methodology is to prune a set of periodicities $P_i$ for a given critical/trusted task such that the system, in closed loop with this control task, delivers its desired performance (see Fig.~\ref{fig: Overview_figure}).
In the {\em second} step, we intelligently select a subset of secure sampling periods from this performance-aware set of periodicities, such that the attacker's schedule analysis is maximally hampered. We theoretically prove that the devised secure periodicity selection policy maximizes the non-determinism in the execution pattern of the victim control task, by preventing the attacker from inferring critical task parameters. The {\em third} step uses these performance-aware and secure set of sampling rates to generate all possible valid task schedules by following the schedule randomization strategy given in~\cite{yoon2016taskshuffler}.
These \emph{valid} schedules are 
 generated by ensuring that no task misses their deadlines while switching between the periodicities. In the {\em fourth} step, we analyse the probability of posterior schedule-based attack (SBA) on each victim task for each schedule. By quantifying the vulnerability of each schedule, we rank and suitably store the schedules for using them in runtime to prevent posterior SBAs.

During run-time, MAARS deploys a secure schedule selection algorithm. As discussed earlier, each control task is implemented with a windowed $\chi^2$-detector that flags any attack attempt on its corresponding closed-loop ({\em box with black-dotted outline} in Fig.~\ref{fig: Overview_figure}, refer to Sec.~\ref{sec:detector_model}). Under no attack, MAARS randomly selects and deploys different schedules from a designer-selected set of valid schedules with minimum vulnerability. We term this mode as \emph{Normal Mode} ({\em green-dotted outlined} box in Fig.~\ref{fig: Overview_figure}). MAARS switches to an \emph{Alert Mode} when the detector flags an attack on a certain victim task $\tau_i\in \Gamma_t$({\em red-outlined} box in Fig.~\ref{fig: Overview_figure}). In this mode, the schedule selection algorithm switches to the least vulnerable schedule for victim task $\tau_i$ to ensure minimum posterior attack success. In the following sections, we discuss each design time step and run-time algorithm in detail.
%
%
\vspace{-1.5mm}
\subsection {Performance-aware multi-rate Period Selection}\label{subsecClf}
To ensure a desired performance while arbitrarily switching between the closed loops with a {\em schedulable} sampling periods to reduce the \emph{inferability} of victim task parameters, we must constrain our switching choices.
Let us consider, for the augmented closed loop system (as mentioned in Eq.~\ref{eq:state_space}) corresponding to a control task $\tau_i\in \Gamma_t$, this sampling period switching signal is $\sigma_i:\mathbb{N}\mapsto P_i$, where $P_i=\{p1_i,p2_i,\cdots,pn_i\}$ represents the chosen set of non-zero sampling periods. In presence of this switching signal, we can express Eq.~\ref{eq:state_space} for the control task $\tau_i$ as
 \begin{align}
 \label{Eq:Switched_System}
     X^{(i)}(k + 1) = \mathbb{A}^{(i)}_{\sigma_i(k)}(X^{(i)}(k))
 \end{align}
the switching signal $\sigma_i(k)=pj$ signifies that at $k^{th}$ sampling instance, this augmented system switches to the controller designed with sampling period $pj$. For a switched system like Eq.~\ref{Eq:Switched_System}, each such switchable closed loop systems are considered as its subsystems. As stated in~\cite{liberzon1999basic}, the following condition must be satisfied for maintaining a desired performance while arbitrarily switching between the closed loop subsystems.
\begin{claim}
\label{clmclf}
    In the case of a switched system like Eq.~\ref{Eq:Switched_System}, in order to maintain a desired {\em global uniform exponential stability (GUES)} decay rate of $\gamma$ while arbitrarily switching among different sub-systems, there must exist a common Lyapunov function (CLF) for all its arbitrarily switchable subsystems that satisfy the following set of equations
    \begin{align}
\label{eq: LMIs}
\scriptstyle
    &\kappa_1(\Vert X^{(i)}[k]\Vert)\leq V_i(X^{(i)}[k])\leq \kappa_2(\Vert X^{(i)}[k]\Vert)\\\nonumber
    &\triangle V_i(X^{(i)}[k])\leq \alpha_i V_i((X^{(i)}[k]),\text{\small for }\alpha_i=g(\mathbb{A}^{(i)}_{pj})\text{\small s.t. }\sigma_i(k)=pj
\end{align}
\end{claim}
\noindent Here, $\kappa_1$,$\kappa_2$ are class $\mathcal{K}_\infty$ functions, and $g$:$\{\mathbb{A}^{(i)}_{p1},\cdots, \mathbb{A}^{(i)}_{pn}\}\mapsto$ $\mathbb{R}\setminus \{0\}$. {\em Global uniform exponential stability (GUES)} with a decay rate $\gamma<0$ signifies that at $k^{th}$ sampling iteration $\Vert X(k)\Vert\leq M e^{\gamma k}~\Vert x[0]\Vert$, where $\Vert.\Vert$ is vector norm and $M >0$. We consider a common quadratic Lyapunov function (CQLF) $V_i(X^{(i)})=X^{(i)} \mathcal{P}^{(i)} {X^{(i)}}^{T}$. We validate the existence of this CQLF for all switchable closed loops by solving the following {\em linear matrix inequalities (LMIs)} and finding out a positive definite matrix $\mathcal{P}^{(i)}>0$ such that for given $\alpha_j,\ {\mathbb{A}^{(i)}_{pj}}^{T}\mathcal{P}^{(i)}\mathbb{A}^{(i)}_{pj}-\mathcal{P}^{(i)}\leq \alpha_{j}\mathcal{P}^{(i)}\  \forall pj\in P_i$. Given a large set of periodicities, we only consider the set of sampling periods for which their corresponding closed loops have a CQLF that follows Claim~\ref{clmclf} for a given performance criteria. This ensures that the desired performance is guaranteed while we devise an arbitrary switching strategy between these closed loops to minimize the vulnerabilities of the task execution schedules against posterior SBAs.
\subsection{Secure Sampling Period Selection and Schedule Generation}
\label{subsec:Secure_Sampling_Period_DIscussion}
We devise a novel periodicity selection policy to further prune the set of performance-aware sampling periodicities for a victim control task to maximally reduce the leakage of its critical task parameters (eg,. arrival times, execution offsets).
As discussed earlier, the attacker is capable of figuring out the victim task's sampling rate by observing the network traffic or characteristics of the physical plant~\cite{chen2019novelsidechannel}. 
The attacker chooses the smallest periodicity of the victim task to construct its schedule ladder as the smallest row length of the ladder increases the cardinality of Attack Arrival Instances ($AAI$) to infer the correct future arrival instances of the victim (refer Fig.~\ref{fig:ladder_analysis}). 
As discussed earlier, an attacker with a longer execution time is a more suitable choice since it increases the chance of getting preempted by a victim instance and, thereby, increasing 
the cardinality of Attack Execution Instances ($AEI$). 
%
\par We defend against such an attacker by suitably choosing a set of sampling periods for the higher-priority victim tasks so that the chances of preemption are reduced (recall that an attacker infers the timing information of a victim task based on its preemption by the victim). To quantify the \emph{inferability} of victim task arrivals, considering an attacker task in a schedule, we introduce a metric \emph{inferability ratio} that represents the number of preemptions by the victim task. 

\begin{definition}(Inferability Ratio)
    The Inferability Ratio (IR) of a schedule $s$ for a victim task, $\tau_i$ having minimum sampling rate $p_i$ and an attacker task $\tau_j$, is defined as the ratio between the total number of columns in the schedule ladder where the attacker task $\tau_j$ arrives and gets executed
    and the total number of columns where the attacker task $\tau_j$ arrives. Mathematically,  Inferability Ratio $IR = \frac{\vert AEI\vert \% \vert AAI \vert}{\vert AAI\vert }$, where $\vert \cdot\vert$ denotes cardinality of a set. The numerator uses the modulo operation to reset the numerator when the attacker executes wherever it arrives, indicating no inferability due to no preemptions by higher-priority tasks.
\end{definition} 
If an attacker task is executed in {\em more} number of columns in a schedule ladder designed for a victim task, the IR of that schedule, w.r.t. the victim, and its attacker choice is {\em high}. This limits the number of columns where the attacker task arrives but does not get executed due to preemption by higher-priority victim tasks. Therefore, it is easier for the attacker to guess the arrival column of the victim task. Intelligent selection of sampling rates can reduce the IR by reducing the number of columns where the attacker task is preempted by the victim task, i.e., the attacker arrives but does not execute. 
We claim the following to ensure reduced preemption of the attacker task by suitable selection of victim periodicities: 
%
\begin{claim}
\label{claim:sampling_period_analysis}
Consider an attacker task $\tau_j \in \Gamma_u$ and a victim control task $\tau_i \in \Gamma_t$. For $\tau_i$, let the set of periodicities allowed be $P_i$ with minimum sampling period $p_i = \min \{p \mid p \in P_i \}$. In this scenario, the generated fixed-priority preemptive schedules shall have the least Inferability Ratio (IR) if for each $ p' \in P_i \setminus \{p_i\}$: $p' = n p_i + e_j, n \in \mathbb{N}^+$.
\end{claim}
\begin{proof} We consider that an attacker constructs the schedule ladder with a row length equal to the minimum sampling period $p_i$ of the victim task $\tau_i\in\Gamma_t$ and observes the AAI and AEI of an attacker task $\tau_j\in\Gamma_u$. 
The duration between the attacker task's arrival 
and the subsequent victim task's arrival 
repeats after every $LCM(p_i,p_j)$ unit of time. Note that the $e_j$ time units
are covered by every execution instance of the attacker task. As long as this 
the duration between the attacker's and the subsequent victim's arrival is smaller than $e_j$ 
time units, the attacker task will be preempted by the victim task, and the same preemption will repeat 
after every multiple of $LCM(p_i,p_j)$ in the timeline. 
%
\par Now, let us consider that at $k$-th instance of the victim task, there is a predicted chance of preemption while the victim is running with some $np_i$ periodicity where $n, k\in\mathbb{N}^+$. Therefore, this $k$-th victim instance, that arrives at $knp_i$ time  may preempt the $\lfloor \frac{knp_i}{p_j}\rfloor$-th 
instance of the attacker task. Therefore, for successful preemption, the $k$-th victim task instance must {\em arrive before the execution of  $\lfloor \frac{knp_i}{p_j}\rfloor$-th attacker instance is finished, i.e., $\lfloor \frac{knp_i}{p_j}\rfloor p_j+e_j > knp_i$}
Now, consider changing the periodicity of the victim task to some $p^{\prime}_i$, to avoid this preemption. Therefore $p^{\prime}_i$ must ensure that the 
$k$-th victim instance arrives after the $\lfloor \frac{kp^{\prime}_i}{p_j}\rfloor$-th attacker instance finishes, i.e., $\lfloor \frac{kp^{\prime}_i}{p_j}\rfloor p_j+e_j \leq kp^{\prime}_i $. By substituting $p^{\prime}_i$ with $np_i+e_j$, this equation becomes as follows.
\begin{align*}
\scriptstyle
k(np_i+e_j) &\scriptstyle\geq \lfloor \frac{k(np_i+e_j)}{p_j}\rfloor p_j+e_j = \lfloor \frac{k(np_i)}{p_j}+\frac{k(e_j)}{p_j}\rfloor p_j+e_j\\
&\scriptstyle\geq \lfloor \frac{k(np_i)}{p_j}+ \lfloor \frac{ke_j}{p_j}\rfloor + ke_j \% p_j\rfloor p_j+e_j\text{\scriptsize, $\frac{a}{b}$ = $\lfloor\frac{a}{b}\rfloor + a\%b$, } \\
&\scriptstyle\geq \lfloor \frac{k(np_i)}{p_j}\rfloor p_j+ k e_j + ke_j \% p_j+e_j\text{\scriptsize,  since $ k, e_j, p_j \in \mathbb{N}^+$}\\
&\scriptstyle\geq knp_i-knp_i\%p_j+ke_j-ke_j\%p_j\\
&\scriptstyle\geq k(np_i+e_j)-k(np_i\%p_j-e_j\%p_j)
\end{align*}
This always holds true since $k(np_i\%p_j-e_j\%p_j)>0$, because for $\tau_j$ to be schedulable $p_j>e_j$. Note that, $e_j,\ p_j$ are considered as whole numbers since we express them in multiples of $\delta$, and the analysis assumes $\delta$ as the notion of unit time. Therefore, this choice of periodicity $p^{\prime}_i = np_i+e_j$ ensures no preemption of the attacker by any $k$-th victim instance. This, in turn, proves the claim that scheduling the victim with a set of periodicities $P_i$ such that $\forall p^{\prime}\in P_i\{p_i\}, p^{\prime}_i= np_i+e_j$ ensures the least IR for an attacker task $\tau_j\in\Gamma_u$ in any execution schedule.
%
\end{proof}
\noindent The periodicity selection criteria used in Claim~\ref{claim:sampling_period_analysis} heavily restricts the periodicity choices since higher values of $n\in\mathbb{N}^+$ are often not supported by Claim~\ref{clmclf}. We update the criteria as follows in order to get multiple periodicity choices for a victim task with reduced IR. 
\begin{remark}\label{remark1}
For a victim task $\tau_i$, let the set of periodicities allowed be $P_i$ with minimum sampling period $p_i = \min \{p \mid p \in P_i \}$. In this scenario, the generated fixed-priority preemptive schedules shall have a {\em reduced} Inferability Ratio if for each $ p' \in P_i \setminus \{p_i\}$: $p' = n p_i + k', n \in \mathbb{N}^+$, such that $k'\in [e_j,p_i-1]$ considering an attacker task $\tau_j$ with periodicity $p_j$ and execution time $e_j$. Since, $k'\geq e_j$, it ensures zero preemption against any attacker, i.e., $IR=0$. Since, $k\leq (p_i-1)$, the predicted preemption after every $LCM(p_i,p_j)$ time is not repeated and delayed. This clearly reduces the cardinality of AEI, ensuring the reduction in IR compared to the situation when the victim task was scheduled with $p_i$.
\end{remark}
\noindent We thus prune the performance-aware set of periodicities $P_i$ to make it security-aware. For each victim task $\tau_i$, we eliminate the periodicities from the set $P_i$ that do not satisfy the criteria given in Remark~\ref{remark1}. 


%
%
%
\noindent Once all the possible sets of $P_i$ for each victim control task $\tau_i$ are pruned, we utilise \emph{TaskShuffler} ~\cite{yoon2016taskshuffler}, a schedule randomization technique. 
Note that after previous steps our task set $\Gamma=\Gamma_t \cup \Gamma_u$, has a set of period choices for tasks in $\Gamma_t$. Each control task $\tau_i \in \Gamma_t$ is assigned a performance-aware and secure set of sampling periodicities, i.e. $P_i= \{p_i^1,p_i^2...,p_i^{n_1}\}$. 
For $|\Gamma_t | = q$, the number of task specifications for $\Gamma$ is given by: $|P_1| \times \cdots \times |P_q|$, as the periods of tasks $\in \Gamma_u$ are assumed as unique.  Each task specification is given as input to the \emph{TaskShuffler} algorithm for generating the possible set of randomized as well as feasible schedules where no job instance misses the deadline. We denote the overall set of randomized schedules generated considering all task-specifications by $\mathcal{S}$.

\subsection{Schedule Vulnerability Analysis}
\label{section:schedule_analysis}
%
In order to characterize each \emph{valid}/feasible fixed-priority preemptive schedule's vulnerability level, we first formalise the task schedule $s$ as an array, $s = \{s[1], s[2],\cdots,s[l]\}$, $\forall \ s[j] \in \{0,1,\cdots, N\}$. Here, $l$ is the number of time units (each time unit of size $\delta$, same as the column sizes of the schedule ladder, refer to Sec.~\ref{secAtkexec}) in the schedule hyper-period (the time duration after which a static task schedule repeats itself). The $j^{th}$ element $s[j]$ in this array denotes the priority of the task executed at $j^{th}$ time unit, i.e., $s[j] = i\Rightarrow \tau_i\in\Gamma$ is executed at $\delta j$ time. We use $0$ to denote an idle period in the schedule, and $i \in \{0,1,...N\}$ denotes the fixed priority assigned to the executed task at a certain time unit. For example, $s=\{1,2,1,3\}$ refers to a fixed-priority preemptive schedule for a task set~\ref{tab:motivating_task_set}. Job 1 of task 1 with priority 1 executes first in the execution order, then job 1 of task 2 with priority 2, and so on. 
\par Given any schedule $s\in \mathcal{S}$, executing multiple trusted and untrusted tasks from the task set $\Gamma$, first we count the total number of possible posterior attacks on $s$. Given a victim task $\tau_i$ with an AEW length of $\Omega_i$, Eq.\ref{Eq:Attack_counter} gives a total count of all possible posterior attacks on task $\tau_i$ denoted by $C_p(\tau_i)$ in a schedule $s$ of length $l$. 
\begin{align}
\label{Eq:Attack_counter}
\scriptstyle
    C_p(\tau_i) = \sum_{j=1}^l I\left( (s[j]==i) \land \bigvee\limits_{\forall k \in [1, \Omega_i]} \exists \tau_p\in \Gamma_u, (s[j+k]==p) \right)
\end{align}
\noindent Here, $I(x)$ is an indicator function defined as follows: $I(x)=1$ if $x$ is true, else $I(x)=0$ for a proposition $x$ that returns true/false. The indicator function returns 1 if any job from $\tau_p \in \Gamma_u$ executes in any of the AEW $\Omega_i$ of task $\tau_i\in \Gamma_t$. 
Note that even though the attacker can compromise a single untrusted task, the system designer analyses the vulnerability of a posterior attack on each trusted task $\tau_i\in \Gamma_t$ by assuming every untrusted task $\tau_j\in \Gamma_u$ as a potential attacker.  
%
\par Since sampling rates of all the control tasks $\tau_i \in \Gamma_u$ across different schedules vary, the number of job executions of the task that arrive in each schedule hyper-period will also vary. Therefore, it is necessary to express this posterior attack counts on a trusted task $\tau_i$ for a schedule $s_k\in\mathcal{S}$ in probabilistic terms. 
We denote this as the \emph{Attack Probability}, which is defined as follows:
\begin{definition}(Attack Probability)
\label{Def: attack_probability}
     Given a valid schedule $s_k(\in \mathcal{S})$ generated for a task set $\Gamma$, the Attack Probability $AP_{<\tau_i,s_k>}$ of a victim task $\tau_i \in  \Gamma_t\subseteq\Gamma$ is the ratio between the total number of attacker tasks getting executed within the AEW length $\Omega_i$ of victim control task $\tau_i$,  and the total number of job-arrivals of the task $\tau_i$ in the schedule hyper-period of length $l_k$. Mathematically, $AP_{<\tau_i,s_k>} =  \frac{C_p(\tau_i) \ \times p_i}{l_k}$. 
\end{definition}
%
\noindent For each trusted control task  $\tau_i\in\Gamma_t$, the $AP_{<\tau_i,s_k>}$ is different in a schedule $s_k$.
For each safety-critical control task $\tau_i \in \Gamma_t$, we define its \emph{Tolerable Attack Probability} (TAP) as the maximum attack probability on a schedule $s$ that can be tolerated by the closed-loop system, associated with $\tau_i$ such that the plant operates in a safe and preferable operating region. For a control task $\tau_i$, we denote its TAP by $TAP_{i,s_k}$.
\par In a mixed-criticality system, the system designer chooses the criticality values of the tasks based on their importance to the system's operations and how their compromise would affect the system's overall performance~\cite{lee2021mathsf}. Considering tasks with a higher index have a lower priority value, the $q$ safety-critical tasks $\{ \tau_1, \tau_2,.., \tau_q \}$ are arranged in decreasing order of priorities.  We assign them with their criticality values $\{q, (q-1),\cdots,2, 1\}$ respectively. We normalize the criticality values by their total sum of criticality values to derive a \emph{criticality level} of each task. Therefore, we denote the normalized {\em criticality levels} for a task $\tau_i$ as $cl_i = c_i/ \sum_{j=1}^{k}c_j, \ \forall \tau_j\in \Gamma_t$. Since AP is task-specific, a schedule may have a minimum attack probability associated with one safety-critical task, whereas it might have the highest AP for the other safety-critical tasks. To secure the schedule w.r.t. all victim tasks, we use these criticality levels to decide the APs of which higher criticality tasks should be given more importance while quantifying the level of vulnerability of the overall task execution schedule. We define a \emph{Schedule Vulnerability Index (SVI)} for each schedule that is expressed as a weighted (with criticality levels) average of the APs of all high-criticality tasks.
\begin{definition}
\label{Def:Schedule_Vulnerability_Index} (Schedule Vulnerability Index)
Given a schedule $s_k \in \mathcal{S}$ generated by a task set $\Gamma$, its \emph{Schedule Vulnerability Index}, denoted by the symbol $SVI_k$, is the weighted sum of attack probability $AP_{<\tau_i,s_k>}$, weighted with normalized task criticality levels $c_i$ for each control task $\tau_i$.
Mathematically, we can write 
$SVI_k= \sum\limits_{\forall \tau\in\Gamma_t} AP_{<\tau_i,s_k>}  \times cl_i \ \forall \tau_j\in \Gamma_t$
\end{definition}
Similar to TAP, we intend to find a tolerable upper limit of SVI. We term this as \emph{Schedule Vulnerability Threshold} (SVT). The schedules having SVI above SVT are not {\em normally} (we shall define what we mean by normally in the next section) deployed in order to reduce vulnerability against posterior SBAs. This ensures that the attack probability thresholds of the higher-criticality (here same as priority) trusted tasks are taken more into consideration during this vulnerable schedule elimination process than that of the lower-criticality victim tasks. Therefore, the \emph{Schedule Vulnerability Threshold} denoted by $SVT$ can be determined by replacing the Attack Probability with TAP as $SVT_k = \sum\limits_{i=1}^q TAP_i \times cl_i$.
%
The attack-aware schedules generated using the methodology mentioned in Sec.~\ref{subsec:Secure_Sampling_Period_DIscussion} are pre-computed and stored along with their SVIs and victim task-wise AP values after this vulnerability analysis (See Fig.~\ref{fig: Overview_figure}). This process eliminates the need for complex online computations in real-time, thereby minimizing the runtime overhead associated with schedule randomization. 
All valid fixed-priority schedules are stored in an array $\mathcal{A}$, sorted in the order of their respective $SVI$s. At the $K$-th index, the schedule with $SVT$ value is stored. Another data structure, \emph{Schedule Lookup Table} (denoted with $LUT$),  stores task-wise an array of schedule indices (from $\mathcal{A}$) which show less \emph{attack probabilities} w.r.t. a victim task than its TAP. An array of schedule indices stored in the $j$-th row of $LUT$ is sorted in the increasing order of AP w.r.t. task $\tau_j$. For example, if the $i_1, i_2$-th schedules from $\mathcal{A}$ (i.e., $\mathcal{A}[i_1]=s_{i_1},\ \mathcal{A}[i_2]=s_{i_2}$) are stored at $k, k^\prime$-th positions of $LUT[j]$ ($j$-th row in $LUT$) respectively, with $k< k^\prime$, then $AP_{\langle \tau_j,s_{i_1}\rangle}<AP_{<\tau_j,s_{i_2}>}\leq TAP_{j}$.  
In the next section, we discuss the online counterpart of the MAARS framework that uses these stored data structures after the design time analysis for randomized deployment of task execution schedules in an attack-aware manner.
\subsection{Schedule Selection Algorithm}\label{secAlgo}
\begin{algorithm}[!hb]
    \caption{Runtime Schedule Selection Algorithm}
    \label{algo:Schedule_select_algo}
    \scriptsize
    \begin{algorithmic}[1]
        \Require{Schedule Array $\mathcal{A}$, Schedule vulnerability threshold $SVT$, Lookup table $LUT$, Detector flag $Atkflag$, current schedule $s_k$}
        \Ensure{next attack-aware schedule $s_{k^{\prime}}$}
        \Function{SchedSel}{$\mathcal{A}$, $SVT$, $LUT$, $Atkflag$, $s_k$} \Comment{Function}
            \State $K \gets \Call{IndexOfSVI}{\mathcal{A}, SVT}$ \Comment{Storing schedule index with SVI=SVT}\label{algflaginit}
            \State $k \gets \Call{IndexOf}{ s_k }, k^{\prime}\gets -1$\Comment{current and next schedule indices init}  \label{algdtcupdate}
            \If{$AtkFlag==0$}  \label{algifAtkFalse}
                \While{$k^{\prime} \neq k$} \label{algnosame}
                \State $k^{\prime} \gets rand( )\%K$ \Comment{ random index between [0,K)}\label{algnormalmoderand}
                \EndWhile
                \State $s^{\prime} \gets \mathcal{A}[k^{\prime}]$\Comment{Normal Mode}\label{algnormalmode}
            \ElsIf{$AtkFlag >0$} $i\gets AtkFlag$   \label{algifAtkTrue}\Comment{If attack detected on $\tau_i$}
                \While{$k \neq k^{\prime}$}  \label{algnosame1}
                \State $M \gets \Call{LengthOf}{ LUT[i] }$\Comment{schedule count with AP below $TAP_i$ }\label{algalertmoderand0}
                \State $ki \gets rand( )\% M$, $k'\gets LUT[i][ki]$ \Comment{random index $\in$[0,M-1]}\label{algalertmoderand}
                \EndWhile
                \State $s_{k^{\prime}} \gets \mathcal{A}[LUT[i][k^{\prime}]]$   \Comment{Alert Mode}\label{algalertmode}
            \EndIf
            \State\Return $s_{k^{\prime}}$; \label{algReturn}
        \EndFunction
    \end{algorithmic}
\end{algorithm}
\par 
\noindent The runtime algorithm in the MAARS framework 
is responsible for randomly deploying schedules in an attack-aware fashion. The runtime schedule selection function $\Call{SchedSel}{ }$ in Algo.~\ref{algo:Schedule_select_algo} takes the following inputs, \textbf{(i) } The lookup table $LUT$ and the stored schedule array $\mathcal{A}$, \textbf{(ii) }the $\chi^2$-detector flag $Atkflag$, \textbf{(iii) }the schedule vulnerability threshold $SVT$ and \textbf{(iv) }the schedule deployed in the current hyper-period $s_k$. 
The algorithm starts by storing the schedule index in $\mathcal{A}$, which has the same SVI as the input SVT. We use the $\Call{IndexOfSVI}{ }$ method, which takes the sorted schedule array $\mathcal{A}$ and the $SVT$ as inputs to find out this index and store in $K$ (see line~\ref{algflaginit}). In line~\ref{algdtcupdate}, by using $\Call{IndexOf}{ }$ method we find out the index $k$ of the input schedule $s_k$ that is currently deployed. Another variable, $k'$, is also initialised (with $-1$) to store the index of the next schedule. 
%
The algorithm has two distinct modes of schedule selection and deployment operations depending on the detector flag $Atkflag$. 
\par \noindent \emph{\textbf{(1) Normal mode: }} When there are no attack flags raised, i.e., $Atkflag=0$ denoting no attack on any control task, the algorithm operates in {\em Normal Mode} (lines~\ref{algnosame} to~\ref{algnormalmode}). In this mode, schedules with $SVI$ less than $SVT$ are randomly selected from the array $\mathcal{A}$. Line~\ref{algifAtkFalse} of the algorithm checks if the $AtkFlag$ value is zero. If it is zero, then the algorithm selects a random number in the range $[0,K-1]$ and stores it in $k'$ (see line~\ref{algnormalmoderand}). Since MAARS randomly deploys new non-vulnerable schedules at every hyper-period, the random schedule index $k'$ must not match with the currently deployed schedule index $k$. The algorithm keeps picking another random number until $k\neq k'$ for this purpose.
(see the while loop in line~\ref{algnosame}). The next schedule $s_{k^\prime}$ is selected from the $k'$-th row of the schedule array $\mathcal{A}$ (see line~\ref{algnormalmode}).  
\par\noindent{\textbf{\em (2) Alert Mode: }} If the \emph{$\chi^2$-detector} running inside the control task raises an alarm by detecting an attack ($g[k]>Th\Rightarrow AtkFlag>0$, see Sec.~\ref{sec:detector_model}) on a control task $\tau_i$, the algorithm runs in {\em Alert mode}. In this mode, schedules are selected from the schedule array $\mathcal{A}$ such that the attack probability for the victim control task $\tau_i$ is lower than its tolerable attack probability $TAP_i$. In line~\ref{algifAtkTrue}, if $AtkFlag$ is non-zero, we first store its value in $i$. The value of the $Atkflg$ denotes a potential attack on task $\tau_i$. In line~\ref{algalertmoderand0}, we find out the length of the $i$-th row in $LUT$ and store it in $M$. This is used to find a random number $ki$ in the range $[0,M-1]$ (see line~\ref{algalertmoderand}). This number is then used to get the index of the next schedule $k'$ from the $ki$-th column and $i$-th row of $LUT$. Note that each such schedule index picked from LUT satisfies $AP_{<\tau_i,\mathcal{A}_k'>} < TAP_i$. As stated earlier, this random number selection is repeated if the index of the current schedule $k$ matches with the randomly picked next schedule index $k'$ (see the while loop in line~\ref{algnosame1}). We select the schedule $s_{k^\prime}$ from the $k'(\neq k)$-th row of $\mathcal{A}$ (see  line~\ref{algalertmode}). 
\par Finally, $\Call{SchedSel}{ }$ returns the new schedule $s_k'$ for deployment. Algo~\ref{algo:Schedule_select_algo} is run once every hyper-period in order to select and deploy new schedules randomly. This online schedule selection is done in an attack-informed way. As a result, the vulnerability level of the task execution schedule is minimized while maintaining the system's performance. 
\par \textbf{Overhead Analysis:}
The time complexity of Algo.~\ref{algo:Schedule_select_algo} depends on the complexity of selecting a random number so that it does not match the current value. Since we store the schedules in the sorted orders of SVI and AP in arrays, the algorithm can access them in $O(1)$.  
The total memory cost for storing $\mathcal{A}$ and $LUT$ can also be calculated, assuming that each element(of the arrays) takes $b$ data bytes. Each schedule, therefore, comprises rows equal to the length of its hyper-period, say $l_k$ for the $k$-th schedule. Considering a total of $M$ schedules, the total cost is $(\sum_{k=1}^{M} l_k) \cdot b $ bytes.
$LUT$, is stored as a 2-D array of $q$ rows (number of trusted control tasks) and $M$ columns, the cost is $qM$. Hence, the total memory cost is  $\left[ qM + \left( \sum_{k=1}^{M} l_k \right) \cdot M \right] \cdot b\ \text{\textit {bytes}}$.
In the next section, we demonstrate the efficacy of this runtime schedule selection algorithm guided by the design-time analysis in real-time system setup.
%

\begin{figure*}[!ht]
\vspace{-5mm} 
    \centering
    \begin{subfigure}[b]{0.19\linewidth}
        \centering
        \includegraphics[width=\textwidth,clip]{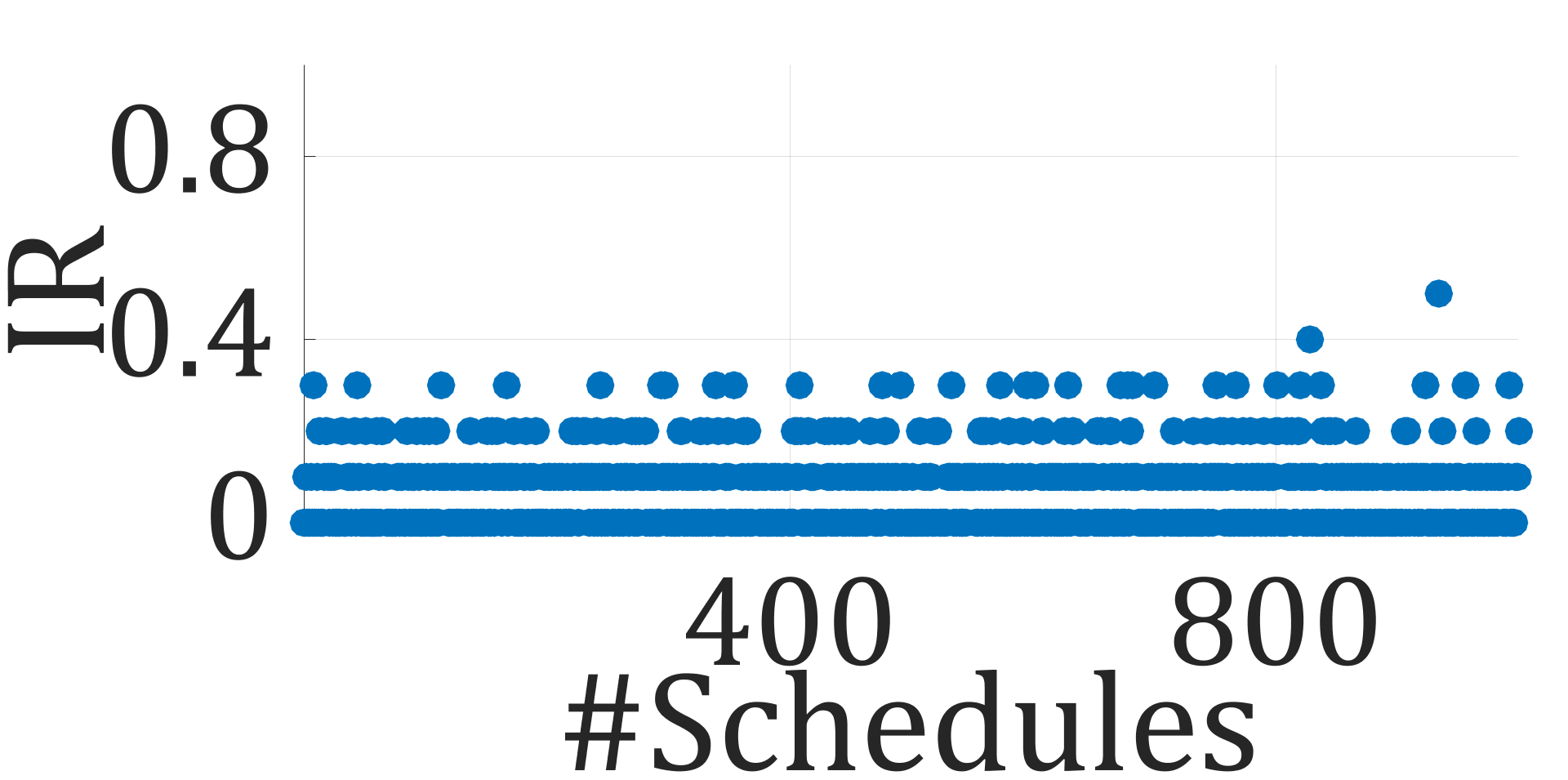}
        \caption{Victim:$\tau_1$, Attack:$\tau_5$}
        \label{fig:maars_lu_v1_a5}
    \end{subfigure}%
    \hfill
    \begin{subfigure}[b]{0.19\linewidth}
        \centering
        \includegraphics[width=\textwidth,clip]{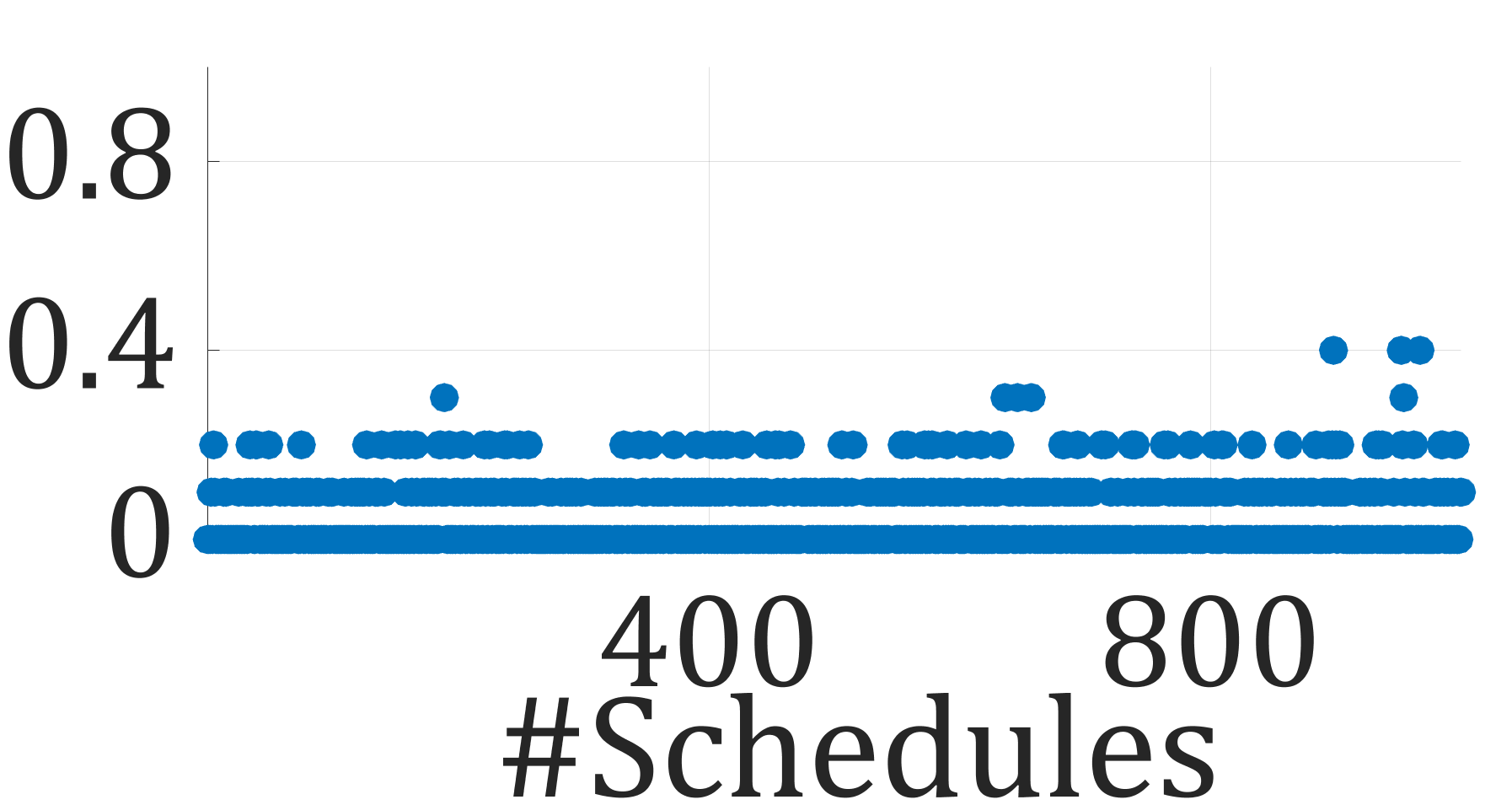}
        \caption{Victim:$\tau_2$, Attack:$\tau_5$}
        \label{maars_lu_v2_a5}
    \end{subfigure}%
    \hfill
    \begin{subfigure}[b]{0.19\linewidth}
        \centering
        \includegraphics[width=\textwidth,clip]{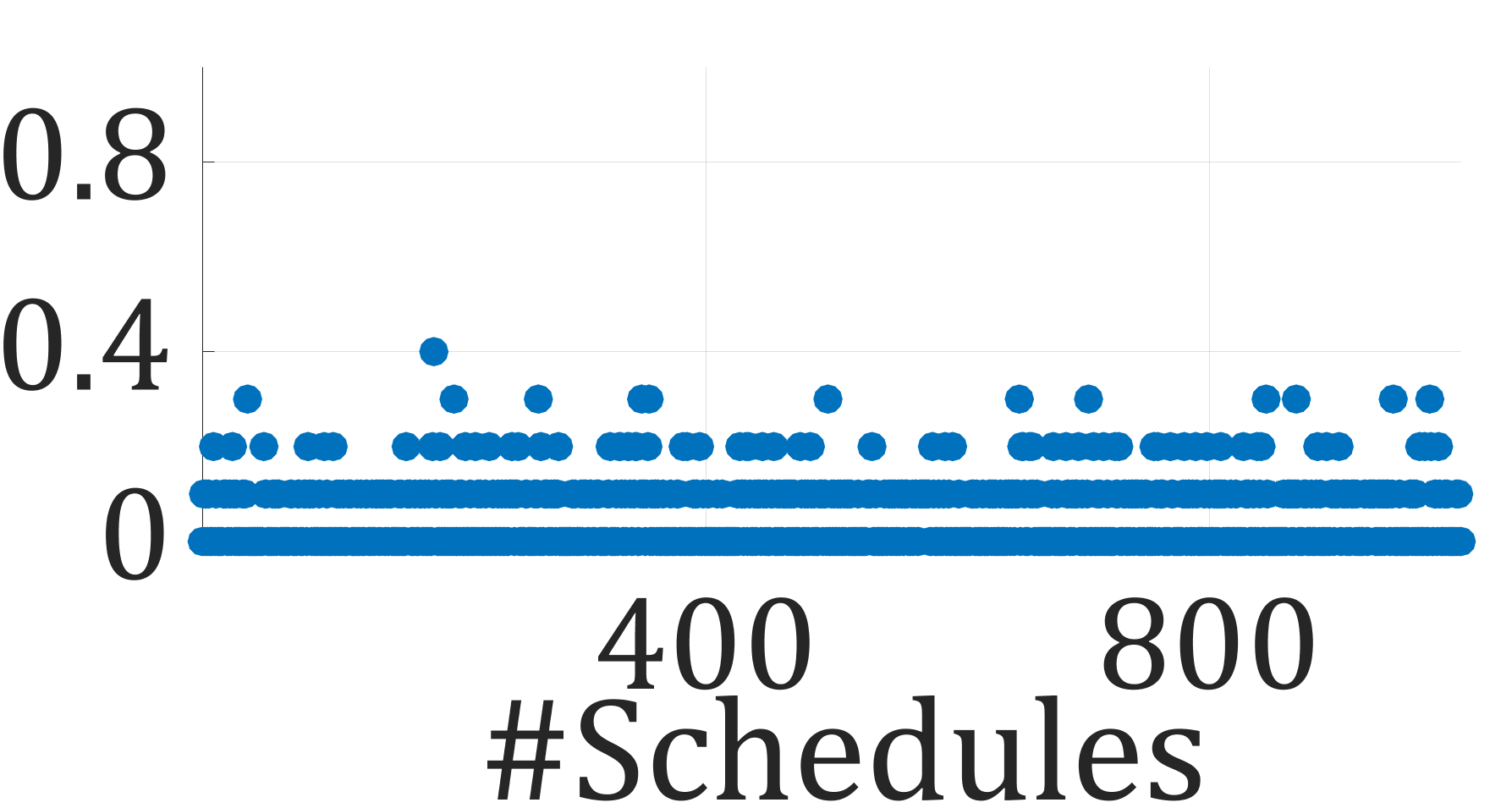}
        \caption{Victim:$\tau_3$, Attack:$\tau_6$}
        \label{maars_lu_v3_a6}
    \end{subfigure}%
    \hfill
    \begin{subfigure}[b]{0.19\linewidth}
        \centering
        \includegraphics[width=\textwidth,clip]{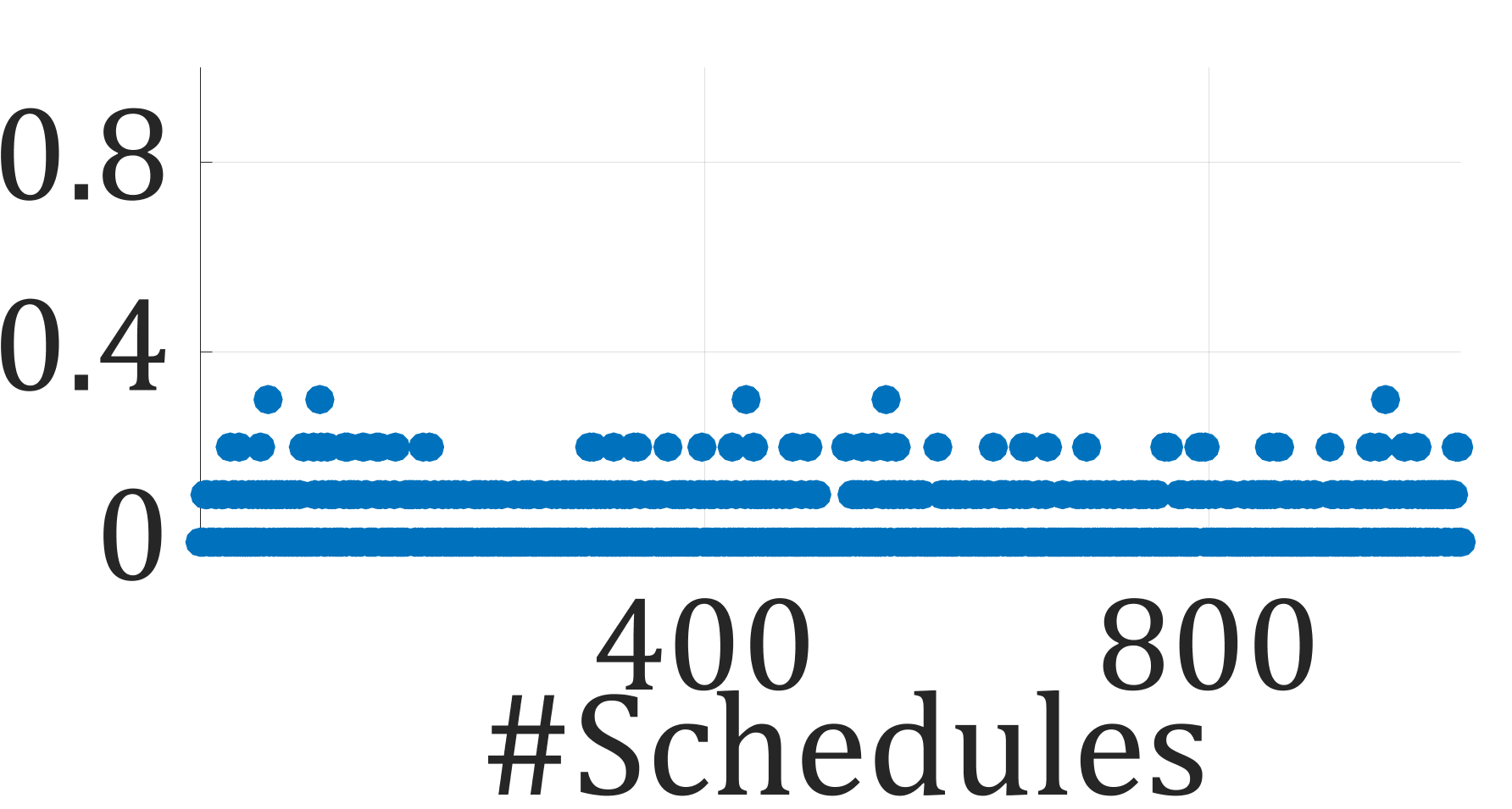}
        \caption{Victim:$\tau_4$, Attack:$\tau_5$}
        \label{fig:maars_lu_v4_a5}
    \end{subfigure}
    \hfill
     \begin{subfigure}[b]{0.19\linewidth}
        \centering
        \includegraphics[width=\textwidth,clip]{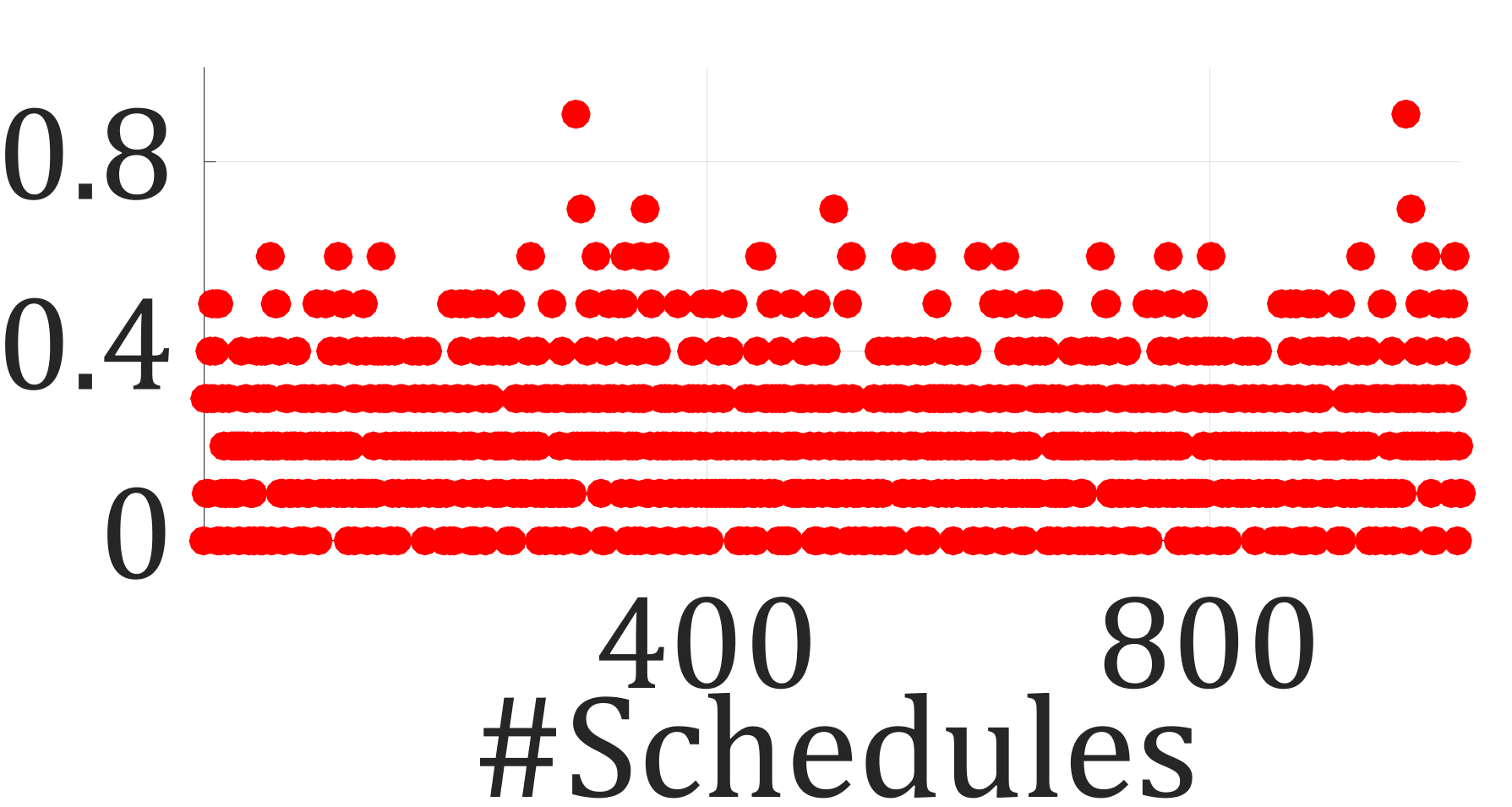}
        \caption{Victim:$\tau_4$, Attack:$\tau_5$}
        \label{fig:taskShuffler_IR_LU}
    \end{subfigure}%
\\
    \begin{subfigure}[b]{0.19\linewidth}
        \centering
        \includegraphics[width=\textwidth,clip]{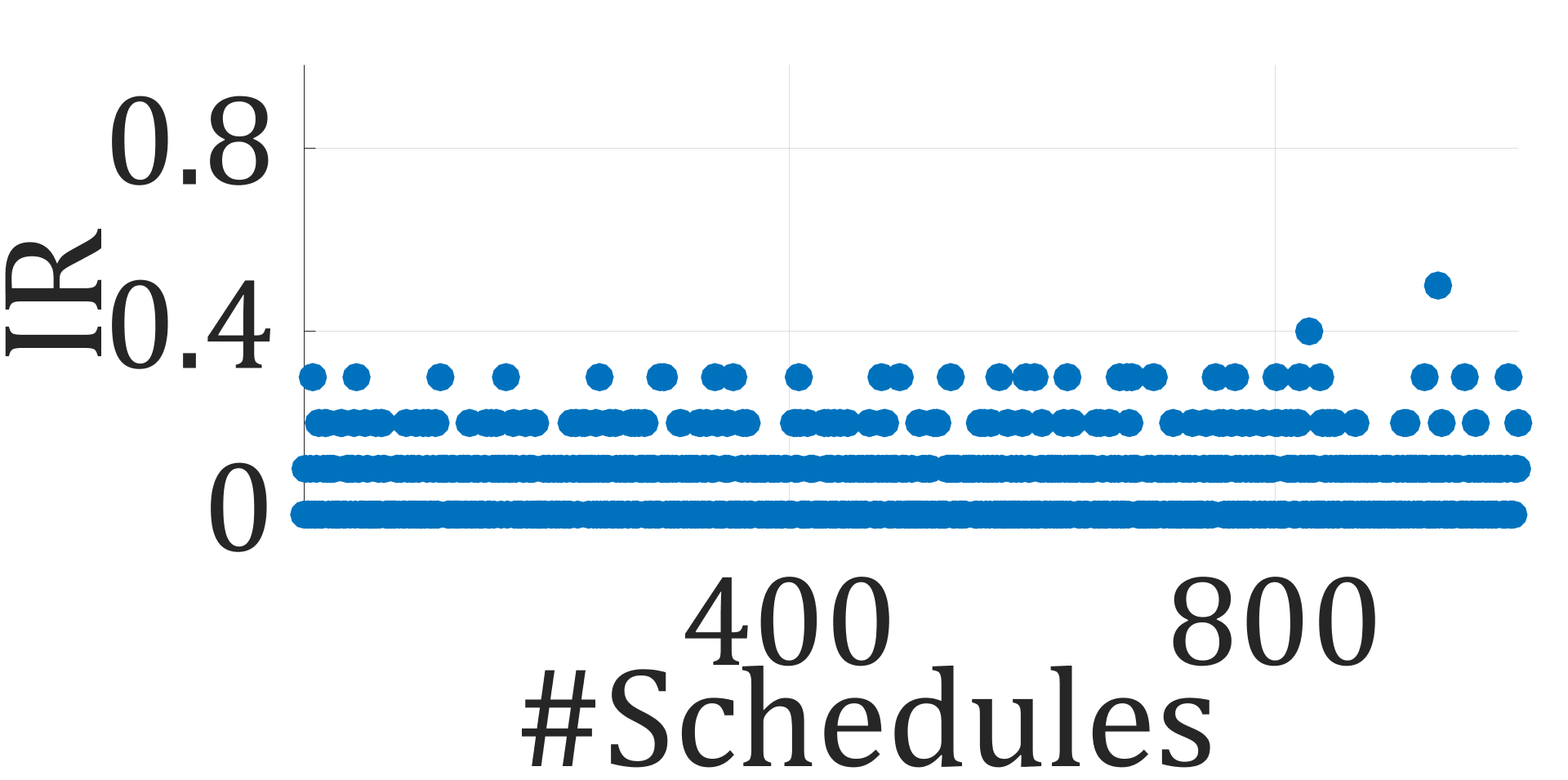}
        \caption{Victim:$\tau_1$, Attack:$\tau_5$}
        \label{subfig:maars_hu_v1}
    \end{subfigure}%
    \hfill
    \begin{subfigure}[b]{0.19\linewidth}
        \centering
        \includegraphics[width=\textwidth,clip]{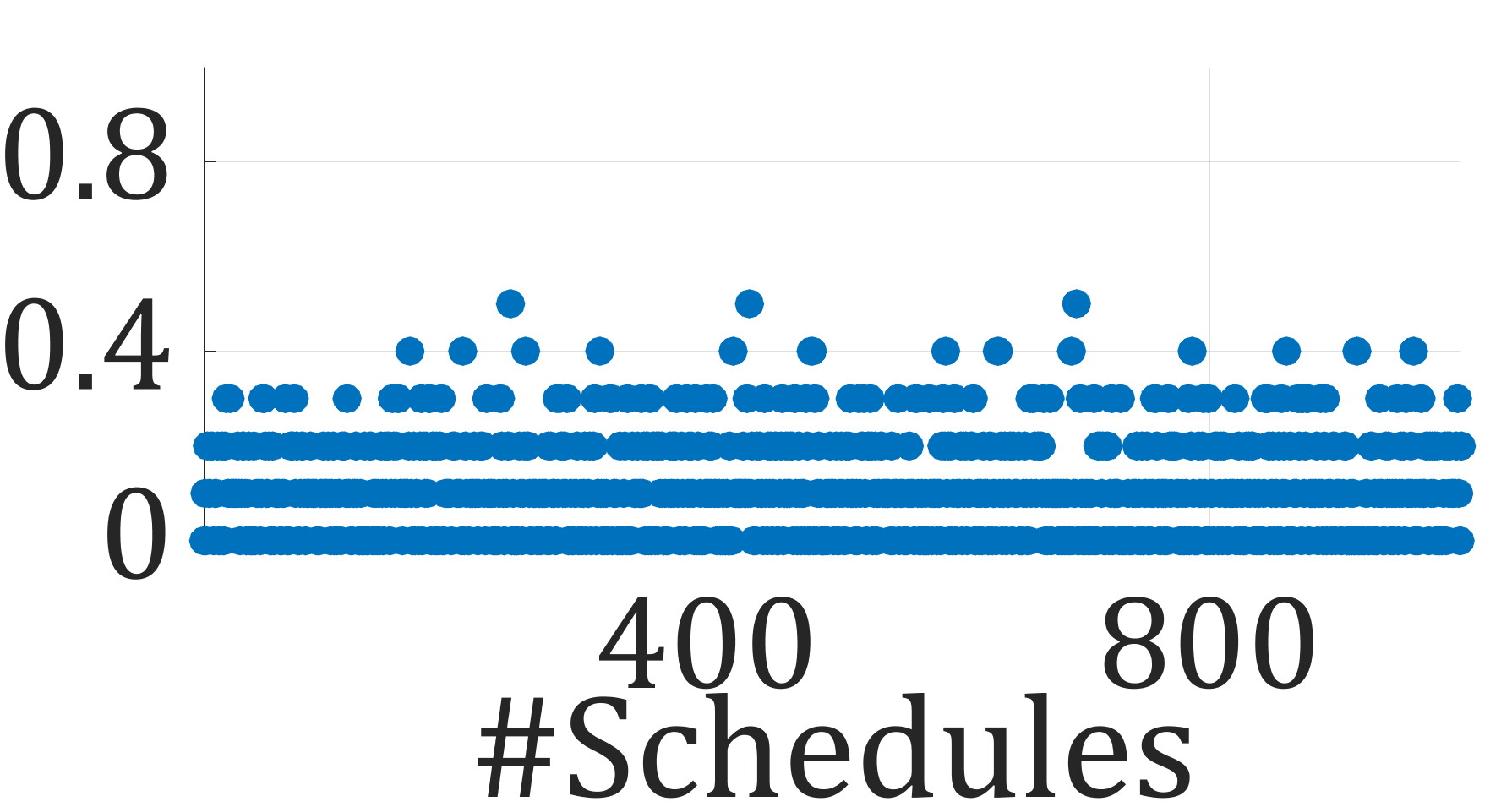}
        \caption{Victim:$\tau_2$, Attack:$\tau_8$}
        \label{subfig:maars_hu_v2}
    \end{subfigure}%
    \hfill
    \begin{subfigure}[b]{0.19\linewidth}
        \centering
        \includegraphics[width=\textwidth,clip]{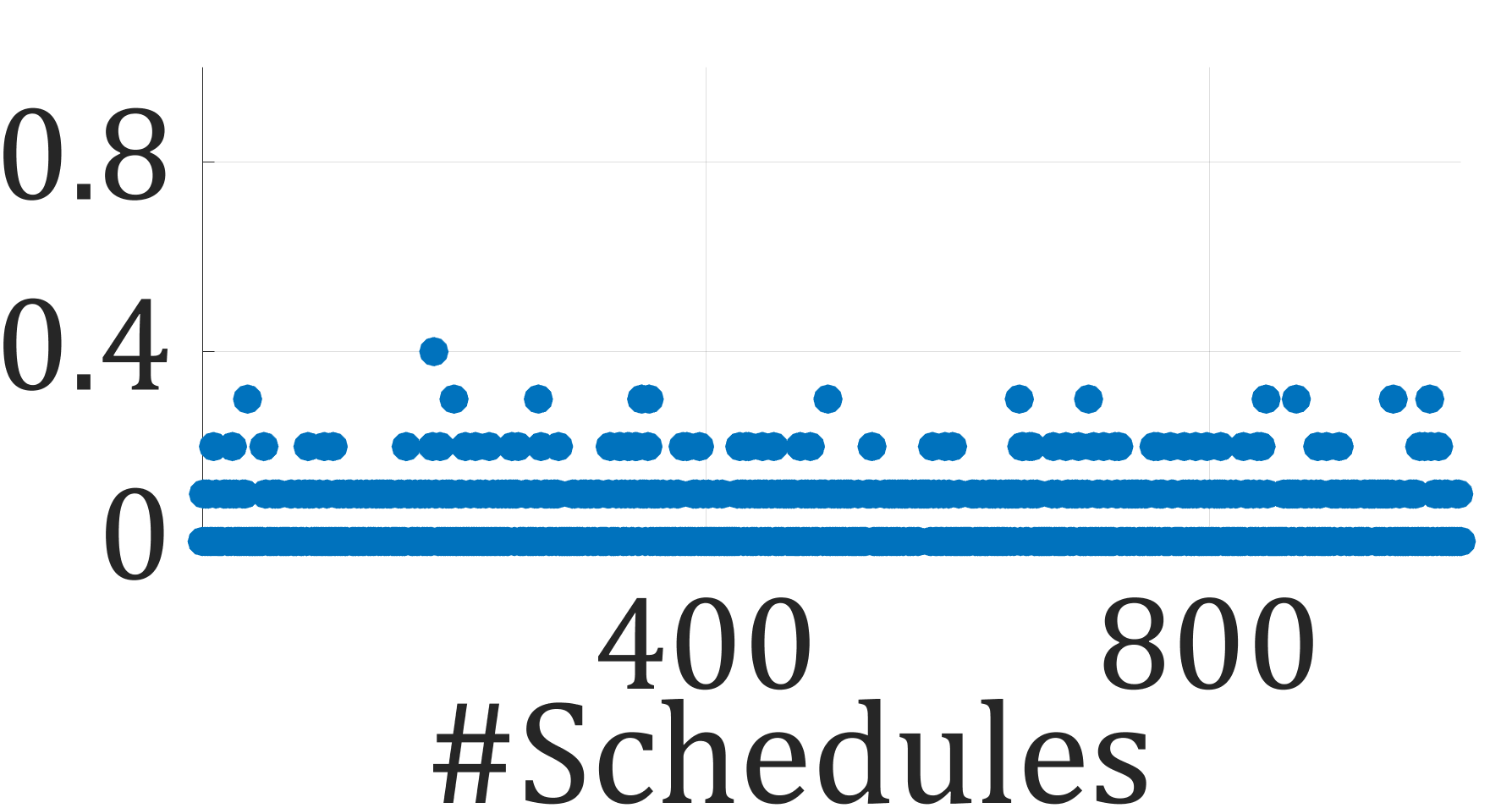}
        \caption{Victim:$\tau_3$, Attack:$\tau_8$}
        \label{subfig:maars_hu_v3}
    \end{subfigure}%
    \hfill
    \begin{subfigure}[b]{0.19\linewidth}
        \centering
        \includegraphics[width=\textwidth,clip]{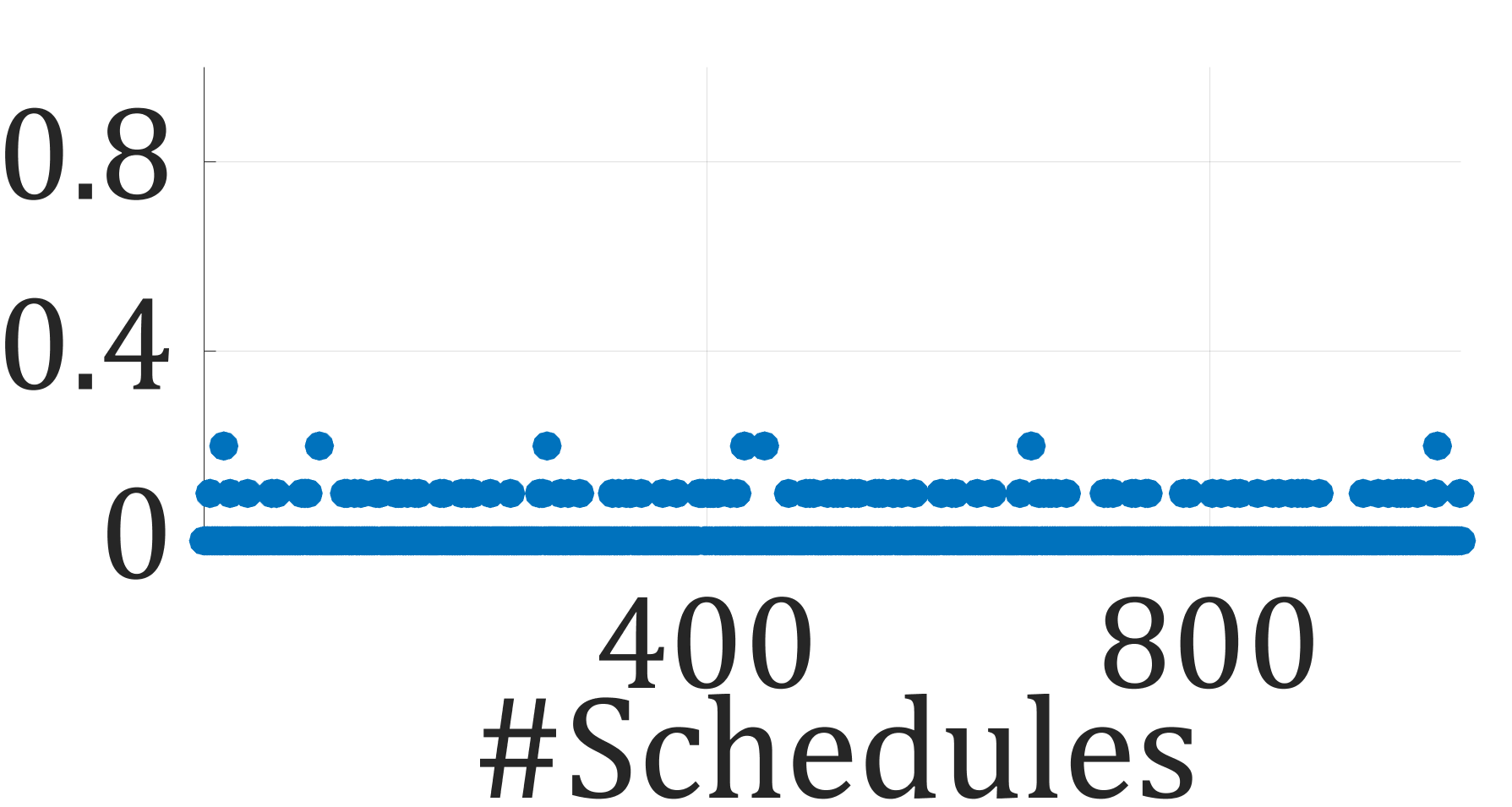}
        \caption{Victim:$\tau_4$, Attack:$\tau_8$}
        \label{subfig:maars_hu_v4}
    \end{subfigure}
    \hfill
     \begin{subfigure}[b]{0.19\linewidth}
        \centering
        \includegraphics[width=\textwidth,clip]{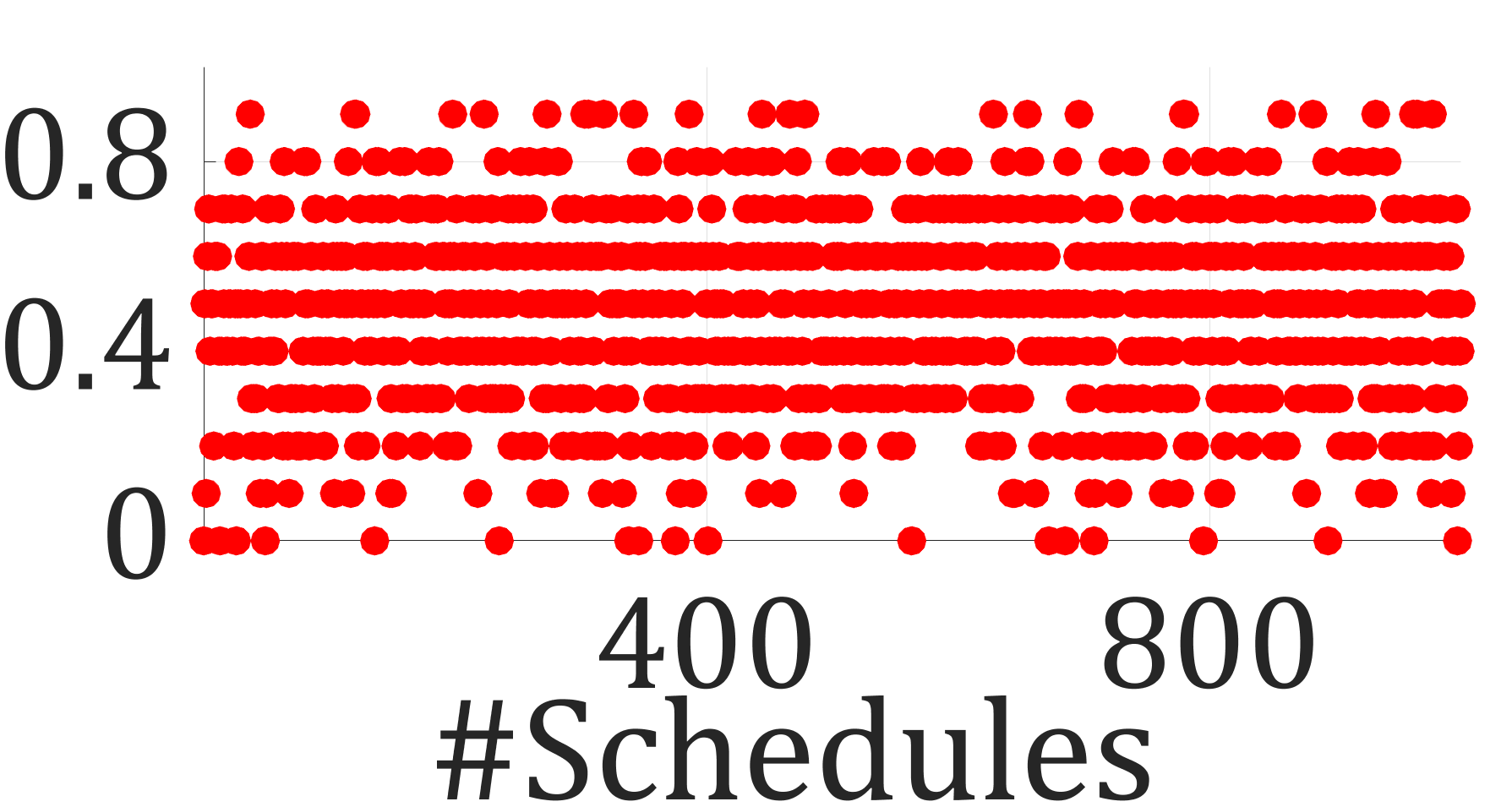}
        \caption{Victim:$\tau_1$, Attack:$\tau_8$}
        \label{fig:taskShuffler_IR_hU}
    \end{subfigure}%
    \caption{\centering  IR in Low Utilization: MAARS (a),(b),(c),(d) and TaskShuffler (e) 
IR in High Utilization Case: MAARS (f),(g),(h),(i) and TaskShuffler (j)}
\vspace{-6mm}
    \label{fig:Inferability_HU}
\end{figure*}
%
\section{Experimental Evaluation}
\label{Experimetal_Evaluation}
\textbf{Experimental Setup:} We experimentally evaluate our proposed \emph{MAARS} framework in a Hardware-in-Loop (HIL) testbed and compare it with state-of-the-art (SOTA) approaches against posterior SBAs. 
The first step of the experiment involves \emph{design-time} schedule generation and analysis for inferability and vulnerability of safety-critical trusted control tasks. This is performed using MATLAB and Python on an 8-core 7th generation Intel i7 CPU with 16 GB of RAM. The implementations can be found in~\cite{gitrepo}. In the next step, we perform a \emph{run-time} analysis of the control tasks under posterior SBA. For this purpose, we employ an ARM-based 32-bit \emph{Infineon Tricore Aurix-397 ECU}, connected to \emph{ETAS Labcar}, a Real-time PC (RTPC), via Controller Area Network (CAN). The implementations can be found in~\cite{gitrepo}.
\begin{table}[h]
\scriptsize
\vspace{-3mm}
\begin{tabular}{|c|c|c|c|c|c|c|c|}
\hline
Task &
  \begin{tabular}[c]{@{}c@{}}Control\\ Loop\end{tabular} &
  \begin{tabular}[c]{@{}c@{}}$P_i$\\(ms)\end{tabular} &
  WCET &
  \begin{tabular}[c]{@{}c@{}}$C_i$\end{tabular} &
  \begin{tabular}[c]{@{}c@{}}TAP\end{tabular} &
  \begin{tabular}[c]{@{}c@{}}$\Omega_i$\end{tabular} &
  \begin{tabular}[c]{@{}c@{}}Schedules\\under TAP\end{tabular} \\ \hline\hline
$\tau_1$& ESP & \{10,15,25\} & 1 & 0.4 & 0.1 & 3 & 15,68 \\ \hline
$\tau_2$& TTC & \{10,15,25\} & 1 & 0.3 & 0.2 & 4 & 122,762 \\ \hline
$\tau_3$& CC  & \{10,25,35\} & 1 & 0.2 & 0.1 & 1 & 72,1088 \\ \hline
$\tau_4$& SC  & \{20,25,30\} & 1 & 0.1 & 0.2 & 8 & 297,4023 \\ \hline
\end{tabular}
\captionsetup{font=scriptsize}
\caption{Trusted Task Set Parameters}
\vspace{-4mm}
\label{tab:experiment_task_set}
\end{table}
\par The physical system/plant is implemented on the RTPC, which sends the sensor measurement to the ECU via CAN, where the control tasks are co-scheduled. Following \emph{AUTOSAR} mandates\cite{furst2009autosar}, we design separate reception tasks to filter and receive sensor IDs transmitted by the ETAS Labcar RTPC over CAN. After processing these sensor data
the control tasks (scheduled by fixed-priority preemptive scheduler in the ECU) are executed. Subsequently, corresponding transmission tasks send computed control input through CAN to the RTPC facilitating plant actuation. Our proposed run-time schedule selection process $\Call{SchedSel}{ }$ 
is scheduled to run once during each hyper-period in the same core as the other control tasks on the ECUs. Note that $\Call{SchedSel}{ }$ doesn't use data from the buffer; rather, it receives the detector signal in an interrupt-driven manner and provides a new schedule to the scheduler for the next hyper-period. So the attacker cannot compromise its operations.
\begin{wraptable}[7]{l}{0.53\columnwidth}
\scriptsize
\vspace{-3mm}
\centering
\begin{tabular}{|l|l|l|l|}
\hline
Task & $p_i$(ms) & $e_i$(ms) & \begin{tabular}[c]{@{}l@{}}Utilization\\ (Low/High)\end{tabular} \\ \hline \hline
$\tau_5$ & 10 & 1 & \multirow{3}{*}{Low, High} \\ \cline{1-3}
$\tau_6$ & 30 & 1 &  \\ \cline{1-3}
$\tau_7$ & 20 & 1 &  \\ \hline
$\tau_8$ & 30 & 5 & \multirow{2}{*}{High} \\ \cline{1-3}
$\tau_9$ & 10 & 2 &  \\ \hline
\end{tabular}
\caption{Experiment: Untrusted Task Set}
\vspace{-6mm}
\label{tab:untrusted_task_table}
\end{wraptable}
We implement four distinct automotive controllers 
similar to the experimental setup in~\cite{koley2021catch,adhikary2022work}. These control tasks include i) an electronic stability program (ESP), responsible for maintaining yaw stability, ii) trajectory tracking control (TTC) system that regulates deviation from a desired longitudinal trajectory, iii) cruise control (CC) system that ensures 
a desired vehicle speed, and iv) suspension control (SC) system that maintains vehicle suspension across varied roads and driving conditions. 
All four control tasks' parameters, their respective AEWs, criticality levels, and TAPs are tabulated in Tab.~\ref{tab:experiment_task_set}. The parameters of the untrusted task set are listed in Tab.~\ref{tab:untrusted_task_table}. We have performed the analysis on two different processor utilization scenarios: Low Utilization (LU) ($\leq 60\%$ with less no. of untrusted tasks, refer to utilization column in Tab.~\ref{tab:untrusted_task_table}) and High utilization (HU) $(\geq 90\%)$. 
\par\noindent\textit{\textbf{(1) Design-time  Analysis: }}
To generate a set of security-aware task schedules, first, it is required to synthesize a set of periodicities for the control tasks (Tab.~\ref{tab:experiment_task_set}) that ensures their control performance. To achieve this, we solve the LMIs given in Claim~\ref{clmclf}
using 
\emph{YALMIP} \cite{lofberg2004yalmip} along with \emph{Gurobi}  solver \cite{gurobi}. 
The performance-preserving set of periodicities of each control task is further pruned following Claim~\ref{claim:sampling_period_analysis} and Remark~\ref{remark1},
to derive the secure sampling rates for all the control tasks. The final pruned set of security-aware performance-preserving periodicities $P_i$ of each control task are presented in the $3$-rd column of Tab.~\ref{tab:experiment_task_set}). 
\begin{figure}[!hb]
    \centering
    \vspace{-5mm}
    \begin{subfigure}[b]{0.49\columnwidth}
        \centering
        \includegraphics[width=\textwidth]{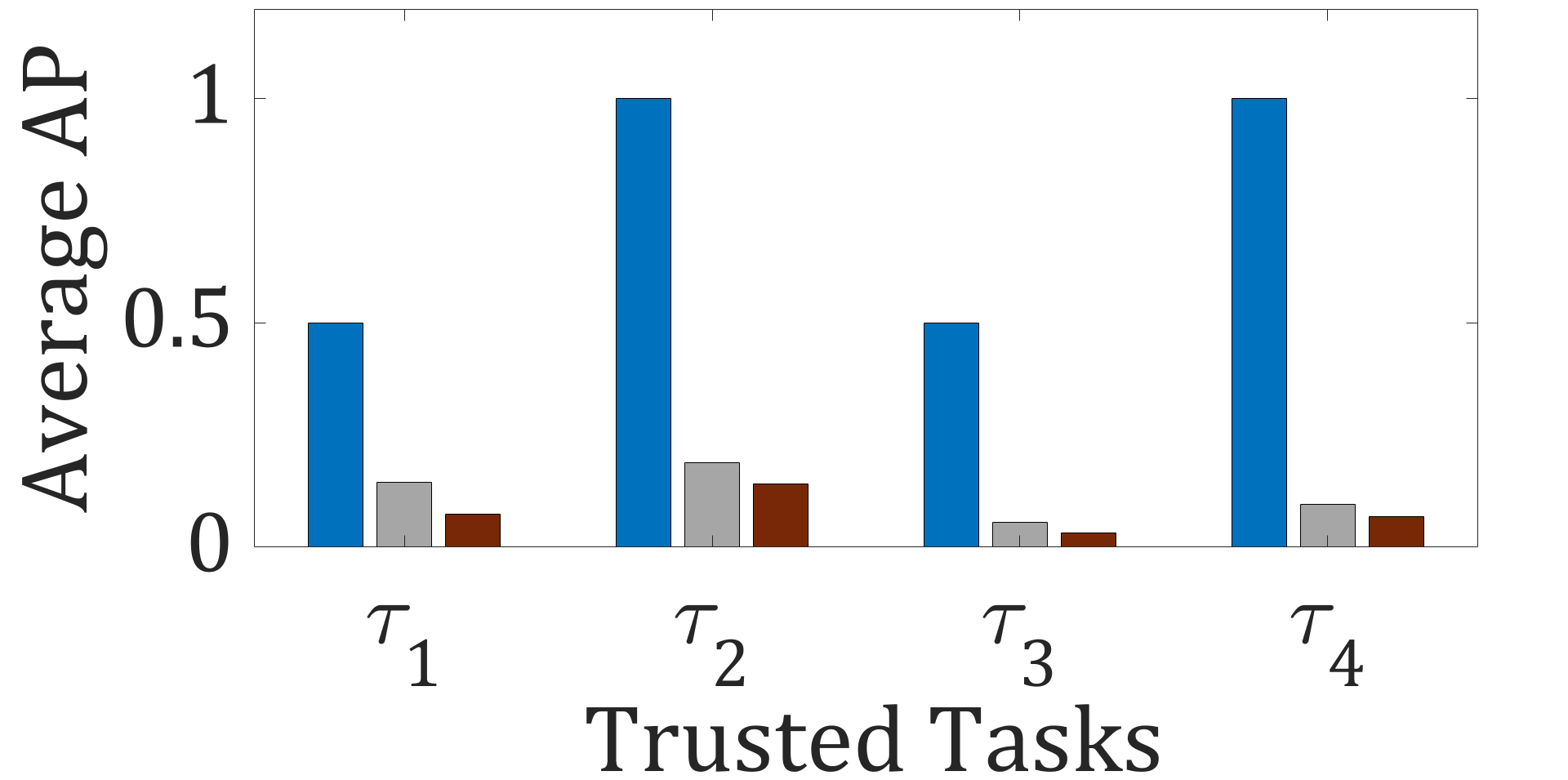}
        \caption{Case 1: Low Utilization}
        \label{fig:LUAP}
    \end{subfigure}
    \hfill
    \begin{subfigure}[b]{0.49\columnwidth}
        \centering
        \includegraphics[width=\textwidth]{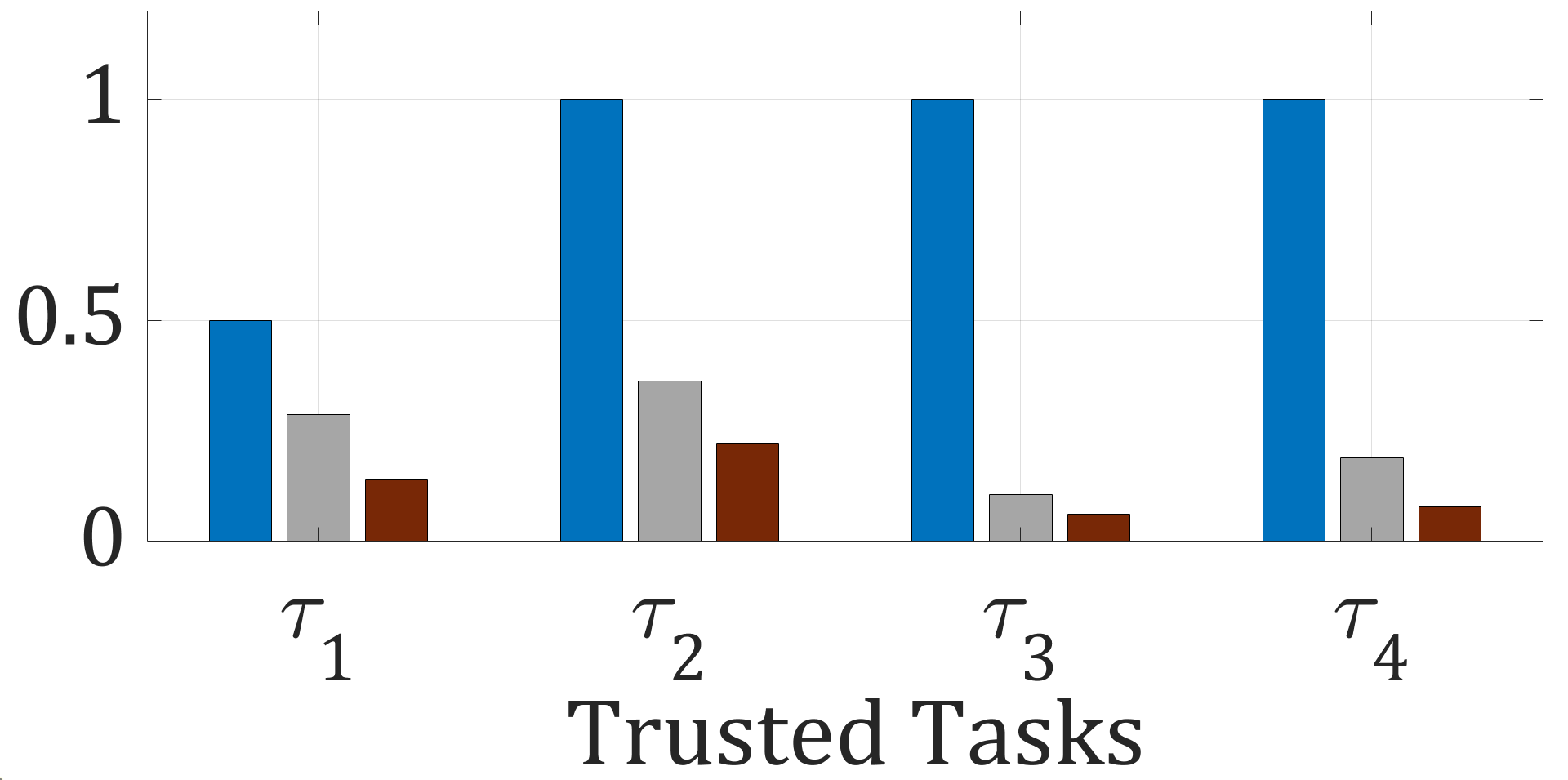}
        \caption{Case 2: High Utilization}
        \label{fig:HUAP}
    \end{subfigure}
      \caption{Average AP: Static Priority Scheduler (Blue), Attack-Unaware Randomization (grey), MAARS (brown)}
      \vspace{-4mm}
    \label{fig:avg_atk_prob}
\end{figure}
Next, we extend the \emph{TaskShuffler} algorithm for multi-rate scheduling, 
as described in Sec.~\ref{subsec:Secure_Sampling_Period_DIscussion} to generate a set $\mathcal{S}$ of valid and randomly deployable schedules using $P_i$'s in both HU and LU scenarios.
For each victim task in Tab.~\ref{tab:experiment_task_set}, we find the number of schedules for all victim tasks with attack probability under their respective TAPs. 
The last column of Tab~\ref{tab:experiment_task_set} shows the number of such schedules. Now, using our experimental results, we discuss and show that our proposed randomized schedule generation method outperforms the SOTA in terms of \emph{inferability ratio}, \emph{attack probability}, and \emph{schedule vulnerability}. As SOTA, we consider two cases: i) processor with \emph{static priority} schedules, and ii) attack-unware randomized schedules are deployed where the set of randomized schedules is generated by \emph{TaskShuffler}, only considering the minimum sampling rates of each control task.
\par\noindent {\em{\bf(A)} Inferability Ratio (IR): }
IR of $1,000$ randomized schedules generated by our proposed MAARS framework are presented in Fig.~\ref{fig:maars_lu_v1_a5},~\ref{maars_lu_v2_a5},~\ref{maars_lu_v3_a6},~\ref{fig:maars_lu_v4_a5} for LU and Fig.~\ref{subfig:maars_hu_v1},~\ref{subfig:maars_hu_v2}~\ref{subfig:maars_hu_v3}, ~\ref{subfig:maars_hu_v4} for HU. The above plots depict the IRs of a certain victim-attacker combination among all possible ones that exhibit the maximum IR.
For example, the schedules of $\tau_1,\tau_2$, and $\tau_4$ show highest inferability for the untrusted task $\tau_5$ and schedules of $\tau_3$ show maximum inferability for the untrusted task $\tau_6$ in case of LU.
Fig.~\ref{fig:taskShuffler_IR_LU} and Fig.~\ref{fig:taskShuffler_IR_hU} (the red plots) represents the highest IR among all the victim-attacker combinations when attack-unaware randomized schedules are deployed.
We can observe that in the case of attack-unaware schedule randomization, most of the randomized schedules show higher IR values compared to the schedules generated by MAARS framework. It implies that MAARS is much better at  hiding timing information of safety-critical control tasks than SOTA.
%
\par\noindent {\it{\bf (B)} Average Attack Probability: }
For each schedule in $\mathcal{S}$, we determine the attack probability of each trusted task using Def.~\ref{Def: attack_probability}. For each trusted task $\tau_i$, we compute the attack probability $AP_i^j$ considering the untrusted task $\tau_j$. The average attack probability $\bar{AP_i}= \frac{\sum_{j=1}^{M}AP_{\langle\tau_i,s_j\rangle}}{M}$ is computed for each trusted task $\tau_i$ where $M$ is the total number of schedules generated by MAARS. Fig.~\ref{fig:avg_atk_prob} shows the average AP for each of the four trusted tasks in Tab.~\ref{tab:experiment_task_set} for both LU and HU. The blue, grey, and brown plots in Fig.~\ref{fig:avg_atk_prob} represent the average APs in case of the schedules generated by static priority schedule, attack-unaware randomized schedule, and the MAARS framework.
We can observe in Fig.~\ref{fig:LUAP}, under LU, schedules generated by MAARS (dark brown plot) show $88\%$ and $38.5\%$ reduction in average AP compared to the static (blue) and SOTA (light grey) respectively. Fig.~\ref{fig:HUAP} shows that under HU, schedules generated by MAARS (dark brown plot) show $83.7\%$ and $47.5\%$ reduction in average AP compared to the static (blue) and SOTA (light grey), respectively.
\begin{wrapfigure}[8]{l}{0.45\columnwidth}
    \centering   
    \vspace{-4mm}
\includegraphics[width=0.48\columnwidth,clip]{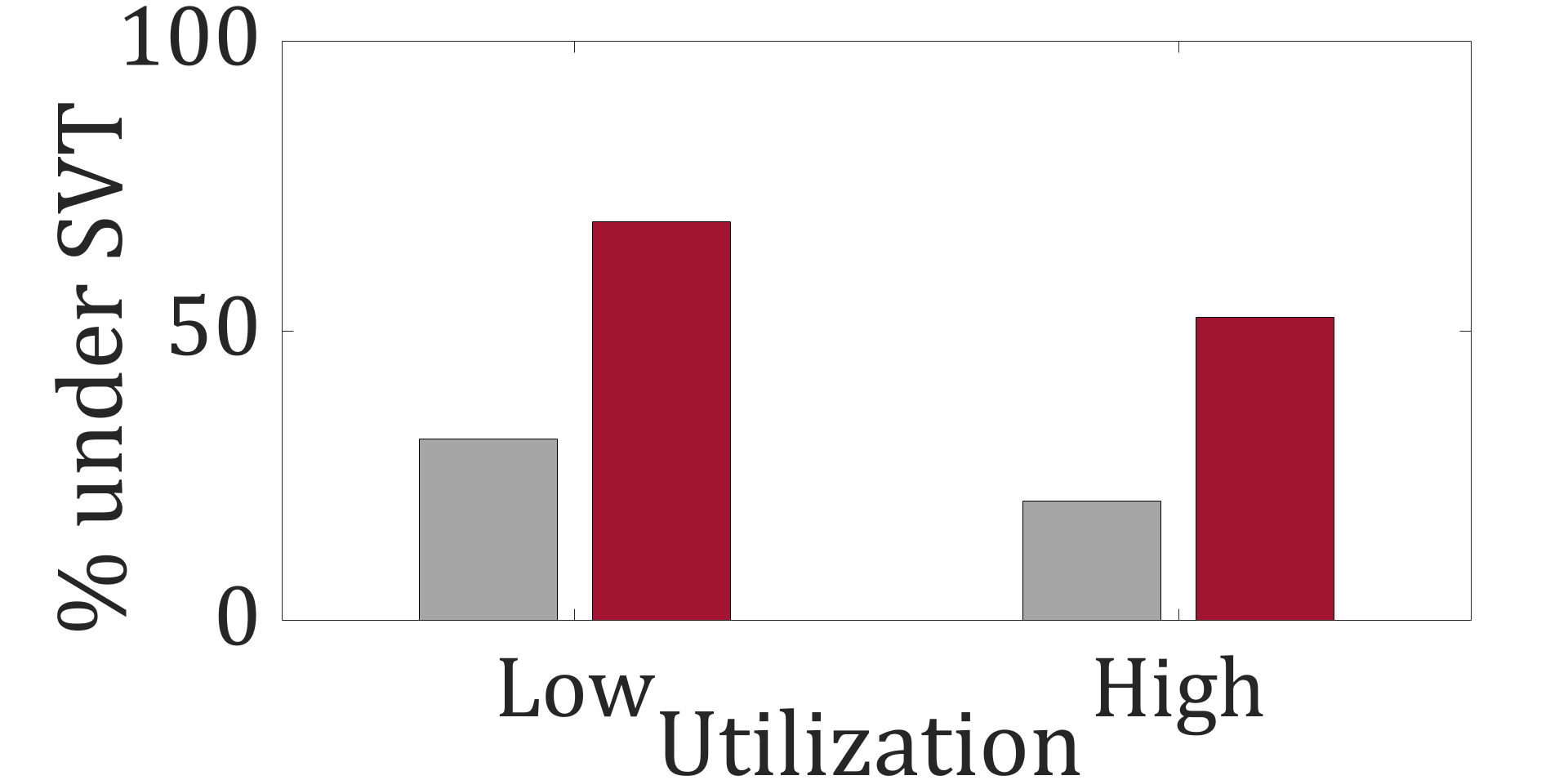}
   \caption{ \% of schedules under SVT: attack-unaware Randomization (grey), MAARS (brown) }
        \label{fig:SVI}
\end{wrapfigure}
\noindent {\it{\bf (C)} Schedule Vulnerability Index: }
In Fig.~\ref{fig:SVI}, we plot the percentage of schedules out of all the schedules generated by MAARS (in brown) and SOTAs (in grey) that have SVIs below the schedule vulnerability threshold or SVT. The SVT value calculated with the TAP values corresponding to our control task set (given in Tab.~\ref{tab:experiment_task_set}) is $0.14$ (refer to Sec.~\ref{section:schedule_analysis}).
Our results show that in the LU scenario, $68\%$ of MAARS schedules have $SVI<SVT$ compared to only $31\%$ for attack-unaware randomization. Similarly, in the HU scenario, our results show that $52\%$ and $20\%$ schedules have SVI below SVT for MAARS and attack-unaware scheduling, respectively. Therefore, by using MAARS, we can generate more schedules with lower vulnerability against posterior SBAs.
\par From the above analysis, we can state that our proposed MAARS framework reduces the inferability ratio and vulnerability index of the schedules and attack probabilities of the tasks compared to the SOTAs.
\par\noindent \textit{\textbf{(2) Runtime Deployment in HIL: }}
Since $\tau_2$ exhibits the highest average attack probability (Fig.~\ref{fig:HUAP}) and $\tau_8$ has the highest schedule inference capability (Fig.~\ref{subfig:maars_hu_v2}) in HU scenario, henceforth we consider $\tau_2$ to be victim task and $\tau_8$ to be the attacker task.
We consider settling time ($t_s$) and GUES decay rate $(\lambda)$ metrics to evaluate the plant's performance. To launch the attack, $\tau_8$ executes a malicious code on the same ECU in which $\tau_2$ is implemented.
The malicious code overrides the control input data sent to ETAS RTPC via CAN. A windowed $\chi^2$-detector (window length=1) is implemented along with $\tau_2$ to detect any unnatural deviation in the plant states. The detector's threshold is set to 4 for the false alarm rate (FAR) to be less than 2\%. We compare the performance of the TTC system under posterior SBA in Fig.~\ref{fig:Attack_on_TTC} when schedules generated by MAARS are deployed and when a static fixed priority schedule is deployed.
In Fig.~\ref{fig:Attack_on_TTC}, the green, blue, red, and magenta plots respectively represent the control input acceleration, plant state deviation from the reference trajectory, the detector's threshold, and $\chi^2$ statistics of the detector. Under no SBA, we note $t_s= 4.3s$ ($<$desired $t_s=10$), and $\lambda =95\%$ when a fixed priority schedule is used in Fig.~\ref{fig: Static_TTC_HIL_no_attack}. However, we can observe in Fig.~\ref{fig:Static_TTC_with_Atk} that the plant becomes unstable due to posterior SBA under the static fixed priority schedule.
On the other hand, Fig.~\ref{fig: MAARS_without_Atk} shows the behavior of the TTC system under no attack when MAARS is deployed.
We can observe $t_s=6s$ and $\lambda=93\%$. Under posterior SBA, the scheduler switches from normal mode to alert mode once $\chi^2$ residue$>4$ (red highlighted area in Fig.~\ref{fig:MAARS_TTC_with_Atk}). In alert mode, we note $t_s=9.2s$ and $\lambda=89\%$. This implies that the system is resilient against posterior attack when MAARS is deployed and subsequently the performance of the system is preserved.
\begin{figure}[!htbp]
    \centering
    \vspace{-6mm}
    \begin{subfigure}[b]{0.49\columnwidth}
        \centering
        \includegraphics[width=\textwidth,clip]{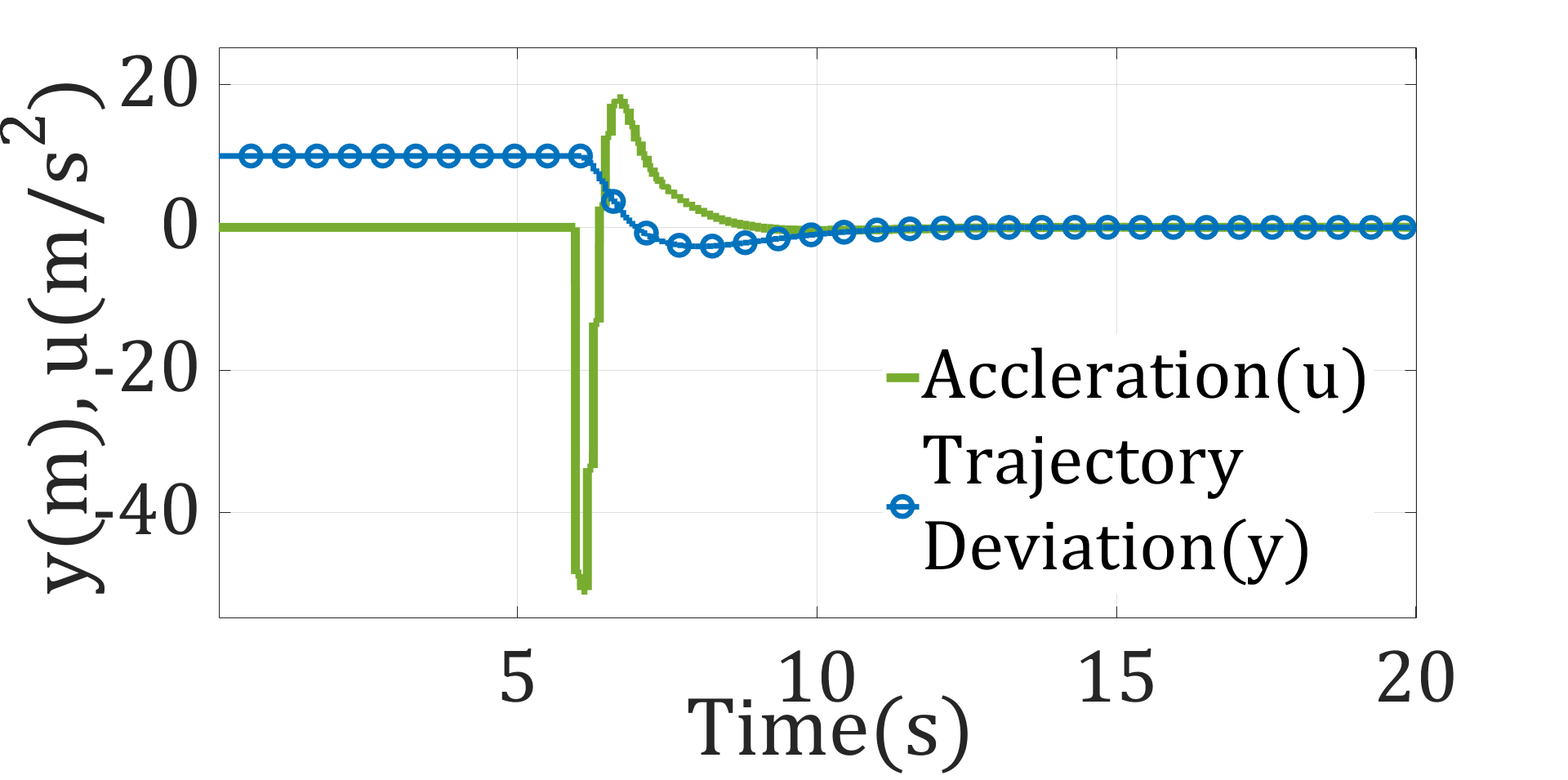}
        \caption{No SBA}
        \label{fig: Static_TTC_HIL_no_attack}
    \end{subfigure}
    \hfill
    \begin{subfigure}[b]{0.49\columnwidth}
        \centering
        \includegraphics[width=\textwidth,clip]{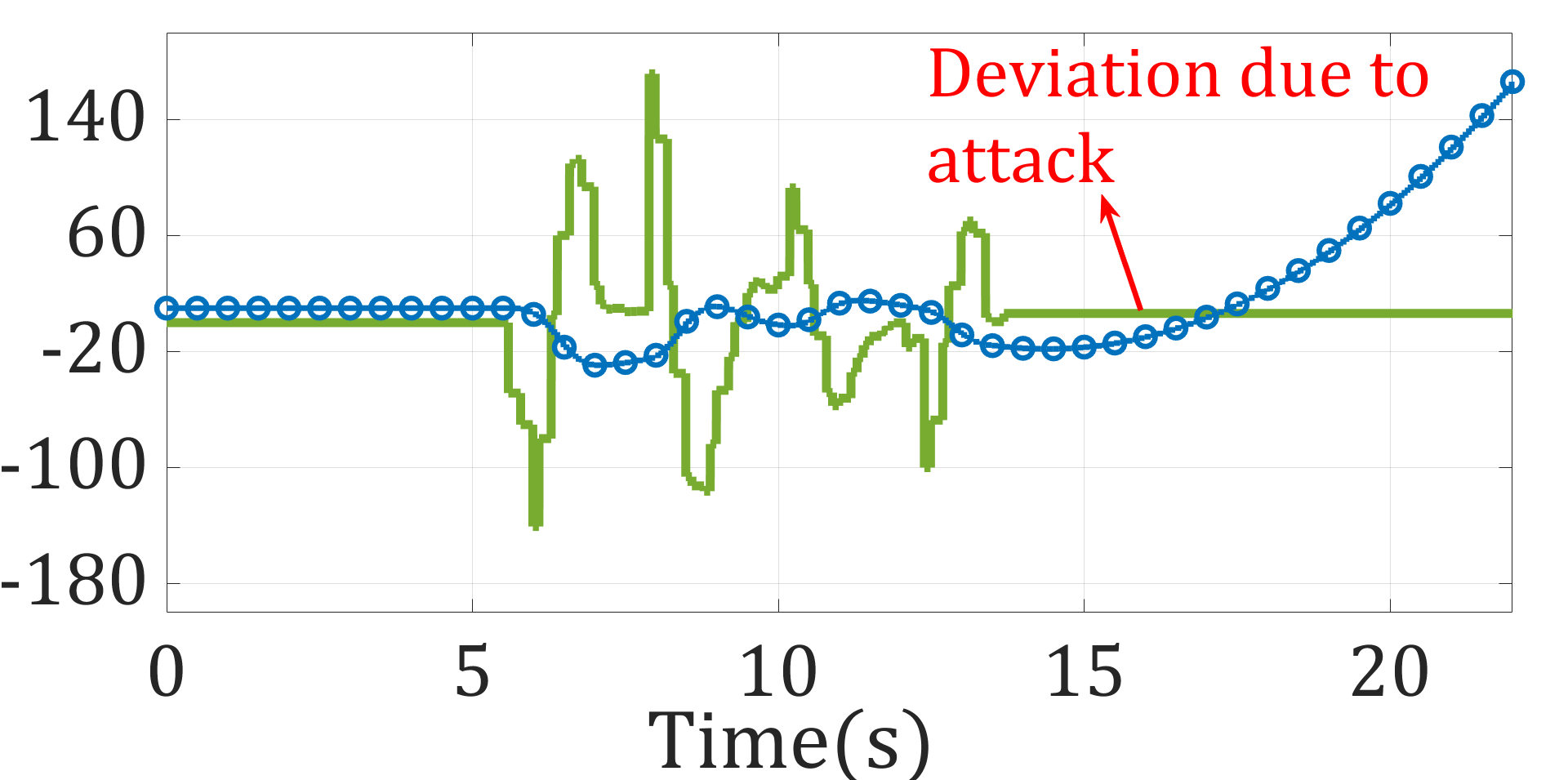}
        \caption{
        With SBA}
        \label{fig:Static_TTC_with_Atk}
    \end{subfigure}
\\
    \begin{subfigure}[b]{0.49\columnwidth}
        \centering
        \includegraphics[width=\linewidth,clip]{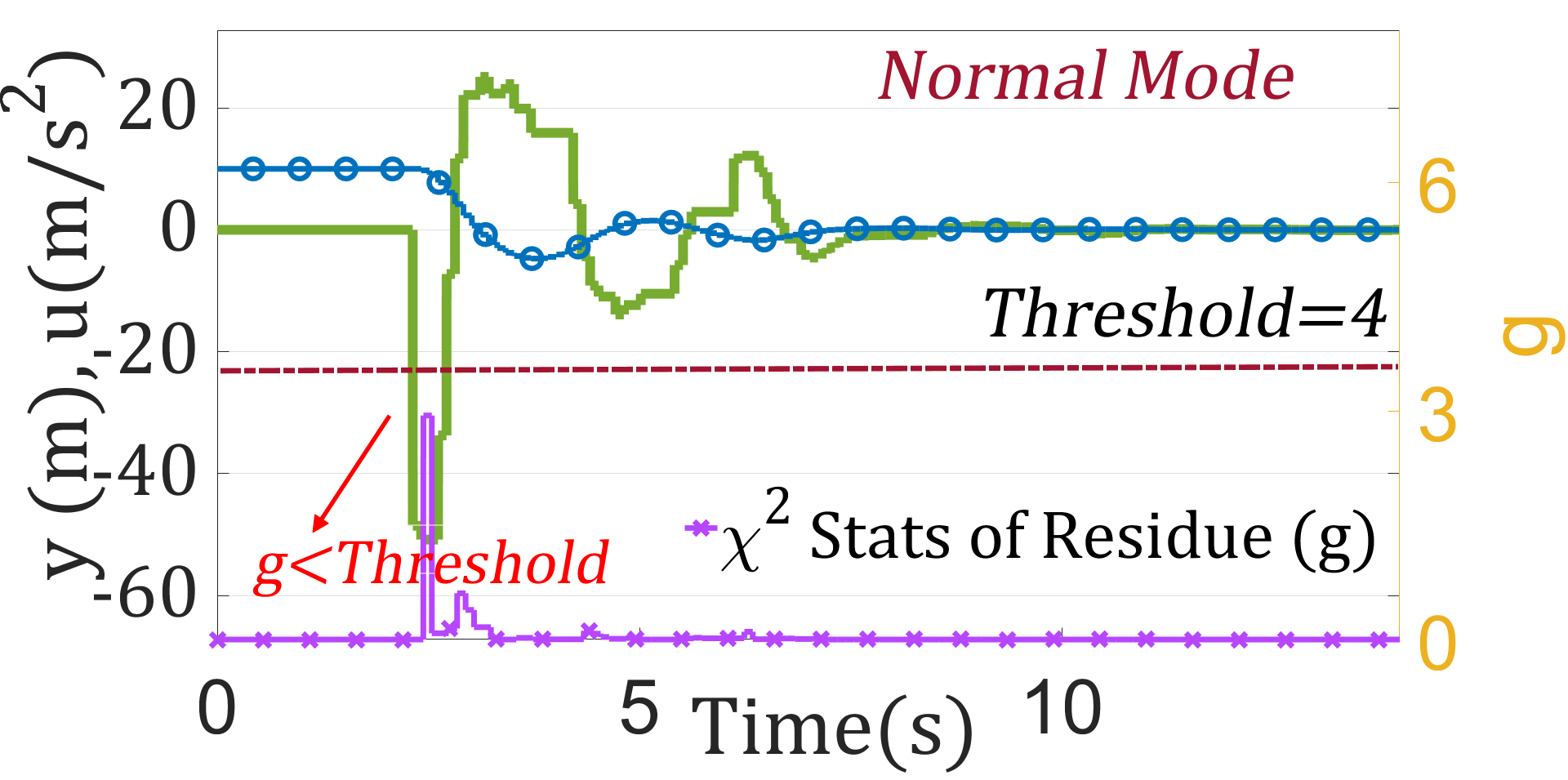}
        \caption{No SBA}
        \label{fig: MAARS_without_Atk}
    \end{subfigure}
    \hfill
    \begin{subfigure}[b]{0.49\columnwidth}
        \centering
        \includegraphics[width=\textwidth,clip]{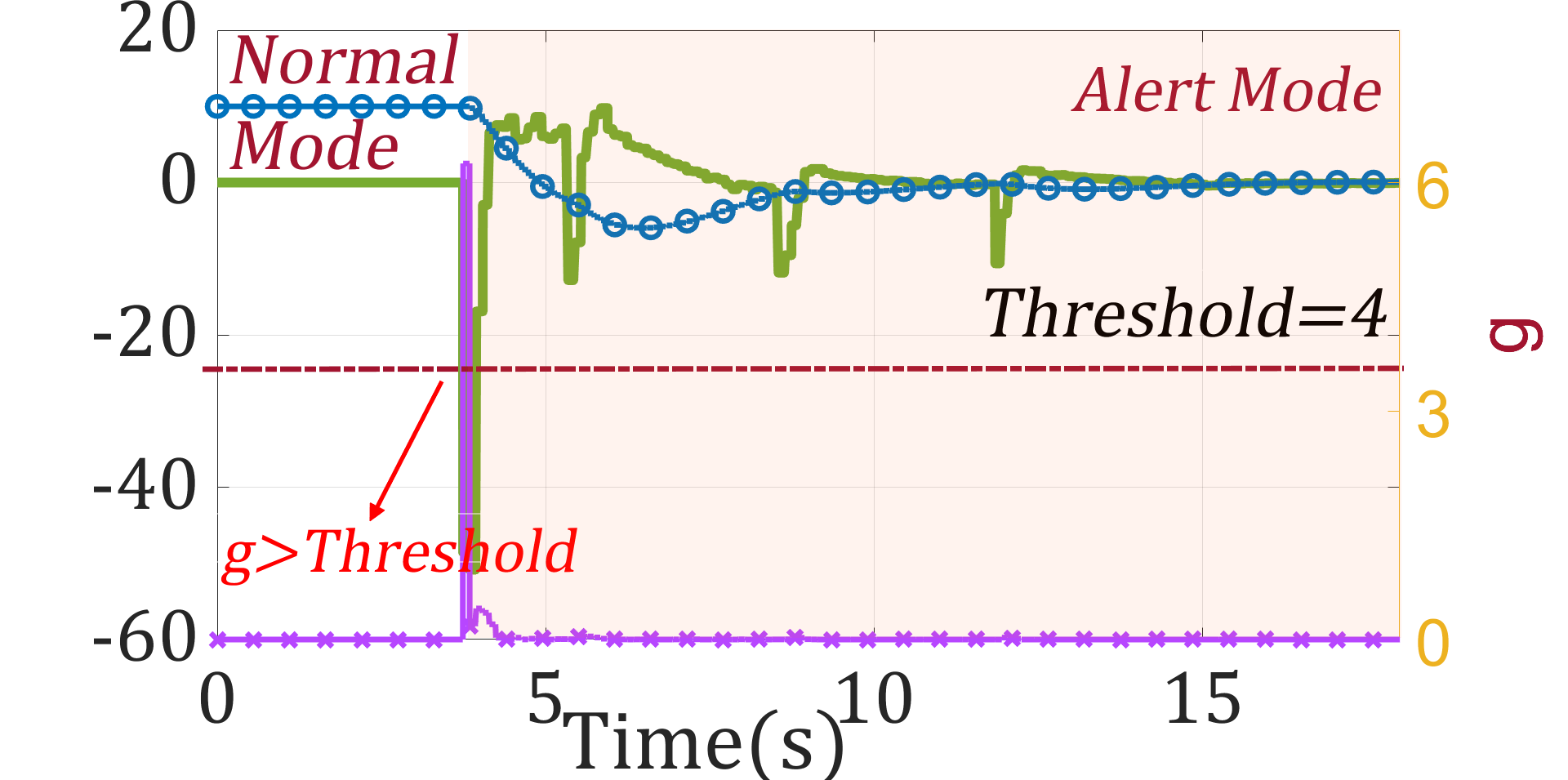}
        \caption{With SBA}
        \label{fig:MAARS_TTC_with_Atk}
    \end{subfigure}
      \caption{\centering Plant Response with TTC: (a),(b) Using static priority scheduler and (c),(d) using MAARS Scheduling Framework}
      \vspace{-5mm}
    \label{fig:Attack_on_TTC}
\end{figure}
\section{Conclusion}\label{secConcl}
Existing schedule randomization methodologies for secureing real-time task-schedules are attack-unaware, hence susceptible to schedule-based attacks. In this paper we present MAARS framework, a novel attack-aware schedule randomization approach that selects performance-aware multiple sampling rates such that the exposure of task parameters is reduced and the vulnerability against posterior schedule-based attacks are mitigated. Our results show that MAARS framework is effective in reducing inferability of critical trusted task parameters, while also reducing posterior attack probability. In future, we aim to extend this work for multiprocessor systems to provide security against other practical schedule-based attacks. 

\bibliography{reference.bib}
\bibliographystyle{ieeetr}
\end{document}